\documentclass[reqno,12pt,letterpaper]{amsart}

\usepackage{amsmath,amssymb,amsthm,graphicx,mathrsfs,url,bm,subfigure,accents,paralist,xargs, slashed,enumitem,mathtools,array}
\usepackage[usenames,dvipsnames]{color}
\usepackage{tikz}
\usetikzlibrary{decorations.markings}
\usepackage[colorlinks=true,linkcolor=Red,citecolor=Green]{hyperref}

\def\?[#1]{\textbf{[#1]}\marginpar{\Large{\textbf{??}}}}

\newlist{inlineroman}{enumerate*}{1}
\setlist[inlineroman]{itemjoin*={{, and }},afterlabel=~,label=\roman*.}

\newcommand{\inlineitem}[1][]{%
	\ifnum\enit@type=\tw@
	{\descriptionlabel{#1}}
	\hspace{\labelsep}
	\else
	\ifnum\enit@type=\z@
	\refstepcounter{\@listctr}\fi
	\quad\@itemlabel\hspace{\labelsep}
	\fi}

\DeclareFontEncoding{FMS}{}{}
\DeclareFontSubstitution{FMS}{futm}{m}{n}
\DeclareFontEncoding{FMX}{}{}
\DeclareFontSubstitution{FMX}{futm}{m}{n}
\DeclareSymbolFont{fouriersymbols}{FMS}{futm}{m}{n}
\DeclareSymbolFont{fourierlargesymbols}{FMX}{futm}{m}{n}
\DeclareMathDelimiter{\VERT}{\mathord}{fouriersymbols}{152}{fourierlargesymbols}{147}

\theoremstyle{plain}
\newtheorem{theo}{Theorem}
\newtheorem{prop}{Proposition}[section]

\newtheorem{lem}[prop]{Lemma}
\newtheorem{cor}{Corollary}

\newtheorem{defi}{Definition}

\theoremstyle{remark}
\newtheorem*{rem}{Remark}
\theoremstyle{definition}

\numberwithin{equation}{section}

\newcommand{\sgn}{\mathrm{sgn}}

\newcommand{\dt}{\tfrac{d\tau}{\tau}}

\DeclareMathOperator{\grad}{grad}
\newcommand{\RR}{\mathbb{R}}
\newcommand{\CC}{\mathbb{C}}

\renewcommand{\Im}{\operatorname{Im}}

\title[The null-geodesic flow near horizons]
{The null-geodesic flow near horizons}

\author{Oran Gannot}
\email{gannot@northwestern.edu}
\address{Department of Mathematics, Lunt Hall, Northwestern University,
	Evanston, CA 60208, USA}

\begin{document}

\begin{abstract}

This note describes the behavior of null-geodesics near nondegenerate Killing horizons in language amenable to the application of a general framework, due to Vasy and Hintz, for the analysis of both linear and nonlinear wave equations. Throughout, the viewpoint of Melrose's b-geometry on a suitable compactification of spacetime at future infinity is adopted.
\end{abstract}

	\maketitle

\section{Introduction}

In this article we study the null-geodesic flow in the neighborhood of a Killing horizon. The surface gravity is given an interpretation as the exponential rate of attraction or repulsion of null geodesics. One novelty is the framing in terms of Melrose's b-geometry on a manifold with boundary. A dramatic example illustrating the utility of this viewpoint can be found in the recent proof of nonlinear stability of slowly rotating Kerr-de Sitter spacetimes \cite{hintz2016global}. In fact, using the language of b-geometry, a general definition of a \emph{b-horizon} is introduced; importantly, this definition applies to spacetimes which are not necessarily stationary.

The primary motivation for this investigation comes from work of Vasy on the microlocal analysis of certain wave equations. In \cite{vasy:2013}, Vasy exploits dynamical properties of the null-geodesic flow on Kerr-de Sitter spacetimes to prove Fredholm properties of the stationary wave operator (see Section \ref{subsect:stationary} below). The same dynamical configuration was also observed near the boundaries of even asymptotically hyperbolic spaces and compactifications of Minkowski spaces \cite[Section 4]{vasy:2013}. 

One reason for defining a b-horizon on a Lorentzian manifold with boundary is to clarify the geometric structure responsible for the similarities between the aforementioned examples. Such a mechanism is not explicitly identified in \cite{vasy:2013} or \cite{hintz2013semilinear}, where each of the examples is treated in an ad hoc manner. Of particular interest for rotating black holes is the role played by an additional symmetry, see Section \ref{subsect:generalhorizon}. Abstracting these common features obviates the need to verify the local hypotheses of \cite[Section 2]{vasy:2013} or \cite{hintz2013semilinear} in individual cases.

In Section \ref{subsect:stationary} below we describe applications to the quasinormal modes of stationary spacetimes bounded by Killing horizons, using the results of \cite{vasy:2013}. In particular, the surface gravity plays an important role in the radial point estimates of \cite[Section 2.4]{vasy:2013} and \cite[Section 2.1]{hintz2013semilinear}. It was observed a posteriori that the surface gravity appears as a natural quantity in these estimates, and here we clarify its role (we remark that the appearance of the surface gravity is more transparent in the work of Warnick \cite{warnick:2015:cmp} on quasinormal modes, see the end of Section \ref{subsect:stationary})

There are also consequences for the study of nonlinear wave equations \cite{baskin2015asymptotics,hintz2013global,hintz2015global,hintz2013semilinear,hintz2016global}, although any reasonable discussion of these works exceeds the scope of this note.

\subsection{Statement of the main theorem} \label{subsect:maintheo}
For the reader familiar with the terminology of Lorentzian b-geometry (see Section \ref{sect:prelims} below for the relevant notation), the main dynamical result is stated below.

The basic setup is as follows: let $M$ be an $(n+1)$ dimensional manifold with boundary $X = \partial M$, equipped with a Lorentzian b-metric $g$ of signature $(1,n)$. Assume that there is a boundary defining function $\tau$ for which $d\tau/\tau$ is timelike on $X$, and orient causal vectors by declaring $d\tau/\tau$ to be \emph{past}-directed. If $G$ is the quadratic form on $^bT^*M$ induced by the metric and $\widehat{\rho}$ is homogeneous of degree $-1$ in the fibers, then the rescaled b-Hamilton vector field $\widehat{\rho} H_G$ on $^bT^*M \setminus 0$ descends to a b-vector field on $^bS^*M$.

We now give a preliminary definition of a b-horizon: let $\mu = g(\tau\partial_\tau,\tau\partial_\tau)$, viewed as a function on $X$, which is independent of the choice of boundary defining function. Thus the sign of $\mu$ determines the causality properties of $\tau\partial_\tau$ on the boundary. By a slight notational abuse the exterior derivative $d\mu \in T^*X$ is considered as an element of $^bT^*X$ (more precisely one should write $\pi (d\mu) \in {^bT^*X}$, where $\pi$ is the injection in \eqref{eq:adjointmap}).

\begin{defi} \label{defi:killing1}
	A subset $H \subset X$ is called a {\bf b-horizon} if $H$ is a compact connected component of $\{ \mu = 0\}$ such that
	\begin{equation} \label{eq:horizoncondition}
	d\mu^\sharp(q) \in {^b N}_qX, \quad q \in H.
	\end{equation}
	From \eqref{eq:horizoncondition}, there exists a function $\varkappa : H \rightarrow \mathbb{R}$, called the {\bf surface gravity} of $H$, such that
	\begin{equation} \label{eq:surfacegravity}
	d\mu^\sharp = 2\varkappa \tau \partial_\tau \text{ on } H.
	\end{equation}
	If $\varkappa$ is nowhere vanishing, then the b-horizon is said to be {\bf nondegenerate}.
\end{defi}

This definition is discussed in detail in Section \ref{subsect:horizon}. A more complete definition that also applies to rotating black hole spacetimes (and for which Theorem \ref{theo:bdynamics} below is also valid) is given in Section \ref{subsect:generalhorizon}. 

In either case, the b-conormal bundle $^bN^*H$ necessarily consists of null b-covectors, hence splits into its future/past-directed components $\mathcal{R}_\pm$. Then, $L_\pm$ is defined to be the image of $\mathcal{R}_\pm \setminus 0$ in $^bS^*_XM$. Near $L_\pm$ a valid choice for $\widehat{\rho}$ is the function $|\tau^{-1}H_G\tau|^{-1} = |2g^{-1}(d\tau/\tau,\cdot)|^{-1}$.

\begin{theo} \label{theo:bdynamics}
If $H$ is a nondegenerate b-horizon, then $L_{\pm(\sgn \varkappa)}$ is invariant under $\widehat{\rho}H_G$, and $L_{\pm(\sgn \varkappa)}$ is a source/sink for the $\widehat{\rho}H_G$ flow within $^bS^*_X M$. Quantitatively, there are nondegenerate quadratic defining functions $\rho_\pm$ for $L_\pm$ within $^bS^*_XM$ such that for each $\delta>0$,
\[
\pm (\sgn \varkappa) \widehat{\rho} H_G \rho_\pm \geq 2(1-\delta) \beta  \rho_\pm 
\]
near $L_\pm$, where $\beta = |H_G \widehat{\rho}\,| > 0$. If $\varkappa$ is constant, then $\beta$ can be replaced by $|\varkappa|$.
\end{theo}

Theorem \ref{theo:bdynamics} shows that $L_\pm$ is either a source or sink for the $\widehat{\rho} H_G$ flow within $^bS^*_XM$, depending on the sign of $\varkappa$. This of course implies the same statement for the null-geodesic flow, namely the flow restricted to the image of $\{G=0\} \setminus 0$ in $^bS^*_XM$. By definition of $\widehat{\rho}$,
\[
\mp \tau^{-1} \widehat{\rho} H_G \tau = 1 \text{ at } L_\pm,
\]
so infinitesimally there is also a stable/unstable direction transverse to $^bS^*_X M$ at $L_\pm$. This shows that if $\varkappa > 0$, then the global dynamics on $^bS^*M$ has a saddle point structure near $L_\pm$, whereas if $\varkappa < 0$, then globally $L_\pm$ is a sink/source; see Figure \ref{fig:globalflow}. The former situation occurs for Kerr-de Sitter spacetimes \cite[Section 3]{hintz2013semilinear}, while the latter occurs for asymptotically Minkowski spacetimes  \cite{baskin2015asymptotics,baskin2016asymptotics}, \cite[Section 5]{hintz2013semilinear}.

	\tikzset{->-/.style={decoration={
			markings,
			mark=at position #1 with {\arrow{>}}},postaction={decorate}}}

\tikzset{-<-/.style={decoration={
			markings,
			mark=at position #1 with {\arrow{<}}},postaction={decorate}}}
\begin{figure} \label{fig:globalflow}
	\centering
	\begin{tikzpicture}
	\draw[->-=.3, very thick] (1,8) -- (4,8);
	\draw[-<-=.7, very thick] (4,8) -- (7,8);
	\draw[fill] (4,8) circle [radius=0.1];
	\node[above] at (4,8) {$L_-$};
	\node[above] at (6.5,8) {${^b}S^*_X M$};
	\node at (5,5.5) {$\varkappa > 0$};
	\draw[-<-=-.6, very thick] (4,5) -- (4,8);
	
	\draw[->-=.5, very thick] (1,7.8) .. controls (3.5,7.6) .. (3.8,5);
	\draw[->-=.5, very thick] (7,7.8) .. controls (4.5,7.6) .. (4.2,5);
	\end{tikzpicture}
	\qquad
	\begin{tikzpicture}
	\draw[->-=.7, very thick] (4,8) -- (1,8);
	\draw[->-=.7, very thick] (4,8) -- (7,8);
	\draw[fill] (4,8) circle [radius=0.1];
	\node[above] at (4,8) {$L_-$};
	\node[above] at (6.5,8) {${^b}S^*_X M$};
	\node at (5,5.5) {$\varkappa < 0$};
	\draw[->-=.7, very thick] (4,8) -- (4,5);
	\draw[->-=.7, very thick] (4,8) -- ++(-135:3);
	\draw[->-=.7, very thick] (4,8) -- ++(-45:3);
	
	\end{tikzpicture}
	\caption{A schematic representation of the null-geodesic flow near $L_-$, where $\varkappa > 0$ on the left, and $\varkappa < 0$ on the right. The horizontal lines represent the flow within the boundary ${^b}S^*_XM$. The directions of the arrows are reversed near $L_+$.}
		\end{figure}
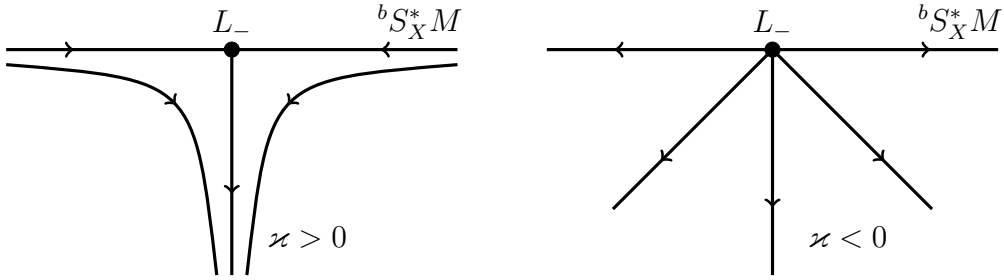

\subsection{Stationary spacetimes and quasinormal modes} \label{subsect:stationary}

Since the language of b-geometry is not yet pervasive in the general relativity literature, we now give a self-contained discussion of the implications of Theorem \ref{theo:bdynamics} in more traditional notation, for the special case of stationary spacetimes. In particular, under additional global dynamical assumptions, there are nontrivial consequences for the study of quasinormal modes --- see Theorem \ref{theo:vasy} below. The relationship between stationary spacetimes and b-geometry is discussed in Section \ref{subsect:relation}.

By a stationary spacetime we mean a Lorentzian manifold $\mathcal{M}$ which admits a complete Killing vector field $T$ transverse to a spacelike hypersurface $\mathcal{X}$, such that each integral curve of $T$ intersects $\mathcal{X}$ exactly once. For the remainder of the introduction assume that $\mathcal{X}$ is compact, possibly with boundary. Translating $\mathcal{X}$ along the integral curves of $T$ for time $t$ induces a diffeomorphism between $\mathcal{M}$ and $\RR_t \times \mathcal{X}$.

For a stationary spacetime, the wave operator $\Box$ commutes with $T$. Therefore the composition 
 \[
 \widehat{\Box}(\sigma) = e^{i\sigma t}\circ \Box \circ e^{-i\sigma t}
 \]
 descends to a well defined differential operator on $\mathcal{X}$. Implicit in this statement is the identification of functions on $\mathcal{X}$ with those on $\mathcal{M}$ which are invariant under $T$. When $\mathcal{M}$ is the domain of outer communcations of a black hole spacetime, the complex poles of $\widehat{\Box}(\sigma)^{-1}$ (with certain outgoing boundary conditions at the horizons) are usually referred to as quasinormal frequencies, otherwise known as scattering poles or resonances. No attempt is made to give an overview of the vast physical and mathematical literature in this subject --- instead, the reader is referred to the reviews \cite{berti2009quasinormal,konoplya:2011,nollert1999quasinormal,zworski2016mathematical}, as well as \cite{zworski:resonances}.
 
 Formally, the distribution of quasinormal frequencies is related to the behavior of solutions to the wave equation by the Fourier inversion formula. In general, the first step in making this correspondence rigorous is to show that $\widehat{\Box}(\sigma)^{-1}$ is well defined and admits a discrete set of complex poles. This can be a substantial task depending on the underlying geometry of $\mathcal{M}$; see \cite{zworski2016mathematical} for some recent developments. Here we focus on a class of stationary spacetimes with a prescribed structure at spatial infinity (in the context of this note, spatial infinity is realized as a boundary component of $\mathcal{M}$). 
 
 Recall that a null hypersurface $\mathcal{H} \subset \mathcal{M}$ is said to be a Killing horizon generated by a Killing vector field $K$ if $\mathcal{H}$ agrees with a connected component of the set where $K$ is null and nonzero. Since $K$ is Killing, $K$ is necessarily tangent to $\mathcal{H}$. In particular, the gradient of $g(K,K)$ and $K$ are orthogonal null vectors on $\mathcal{H}$. Thus there is a function $\varkappa:\mathcal{H} \rightarrow \RR$, called the surface gravity, such that
 \begin{equation} \label{eq:stationarygravity}
\grad_g\left({g(K,K)}\right) = -2\varkappa K
 \end{equation}
 on $\mathcal{H}$. We will consider stationary spacetimes which are bounded by Killing horizons in the following sense (which we remark is not a standard definition).

 \begin{defi}  \label{defi:bounded}
 	A Killing horizon $\mathcal{H} \subset \mathcal{M}$ generated by a Killing vector field $K$ is said to be {\bf compatible} with the pair $(T,\mathcal{X})$ if $[T,K]=0$ and $V=T-K$ is tangent to $\mathcal{X}$.

 	A stationary spacetime $\mathcal{M}$ with boundary is said to be { \bf bounded by nondegenerate Killing horizons} if $\partial\mathcal{M}$ is  a disjoint union of compatible Killing horizons $\mathcal{H}_1,\ldots,\mathcal{H}_N$ generated by Killing vector fields  $K_1,\ldots, K_N$, such that $K_j$ is timelike in a punctured neighborhood $\mathcal{H}_j$, and each surface gravity $\varkappa_j$ is positive.
 \end{defi}

 When $\mathcal{H} \subset \partial M$, either $g(K,K) > 0$ or $g(K,K) < 0$ in a punctured neighborhood of $\mathcal{H}$; we enforce the former alternative. With a view towards quasinormal modes, only Killing horizons with positive surface gravity are considered  in Definition \ref{defi:bounded} (Theorem \ref{theo:bdynamics} applies more generally, regardless of sign).

The fundamental example of a spacetime bounded by Killing horizons is the Kerr--de Sitter solution, corresponding to a rotating black hole with positive cosmological constant. Theorem \ref{theo:bdynamics} was established directly for these spacetimes in \cite[Section 6]{vasy:2013} and \cite[Section 3]{hintz2013semilinear}.  The Kerr and Kerr--AdS spacetimes also contain Killing horizons, but also contain other types of infinities (an asymptotically flat end for Kerr spacetimes and conformally timelike boundary for Kerr--AdS spacetimes).

Let $\mathcal{C}_\pm \subset T\mathcal{M} \setminus 0$ denote the future/past-directed light cones, oriented according to the convention that $dt$ is future-directed. We will focus our attention on those null-geodesics $\gamma$ (with nonzero tangent vectors) which have zero energy with respect to $T$,
\[
 g\left(\dot{\gamma},T\right)=0.
\]
Observe that a zero-energy null-geodesic is necessarily confined to the region where $T$ fails to be timelike. Now if $\mathcal{H}$ is a Killing horizon generated by $K$, decompose its nonzero normal bundle 
\[
N\mathcal{H}\setminus 0 = \left(N\mathcal{H} \cap \mathcal{C}_+\right) \cup \left( N\mathcal{H} \cap \mathcal{C}_-\right).
\] 
The following corollary rephrases the qualitative source/sink statement of Theorem \ref{theo:bdynamics} in terms of geodesics (rather than a flow on the cotangent bundle). In this setting it is not relevant whether $\mathcal{H}$ is a boundary component of $\mathcal{M}$ or embedded in its interior.

\begin{cor} \label{theo:localnontrapping}
	Suppose that $\mathcal{H} \subset \mathcal{M}$ is a compatible Killing horizon with positive surface gravity. Then, there exist conic neighborhoods $U_{\pm} \subset \mathcal{C}_\pm$ of $N\mathcal{H} \cap \mathcal{C}_\pm$ such that if $\gamma$ is a zero-energy null-geodesic and $\dot{\gamma}(0) \in U_{\pm} \setminus \left(N\mathcal{H} \cap \mathcal{C}_\pm\right)$, then $\dot{\gamma}(s)\rightarrow N\mathcal{H} \cap \mathcal{C}_\pm$ as $s\rightarrow \mp \infty$.
\end{cor}

 Corollary \ref{theo:localnontrapping} says nothing about the global behavior of zero-energy null-geodesics. This motivates the following definition.

\begin{defi} \label{defi:nontrapping}
	Let $\mathcal{M}$ be bounded by nondegenerate Killing horizons and $\gamma$ be a zero-energy null-geodesic satisfying $\dot{\gamma}(0) \notin  N\partial \mathcal{M}$. If $\dot{\gamma}(0) \in \mathcal{C}_\pm$, then $\gamma$ is said to be {\bf nontrapped} if the following conditions hold:
	\begin{enumerate} \itemsep6pt 
		\item $\dot{\gamma}(s) \rightarrow N \partial \mathcal{M}\cap \mathcal{C}_\pm$ as $s\rightarrow \mp \infty$,
		\item there exists $s_0 \geq 0$ such that $\gamma(\pm s_0) \in \partial \mathcal{M}$.
			\end{enumerate}
If each zero-energy null-geodesic is nontrapped, then $\mathcal{M}$ is said to be {\bf nontrapping at zero energy}.
	\end{defi}

 \noindent It follows from Theorem \ref{theo:localnontrapping} that if $\mathcal{M}$ is nontrapping at zero energy, then $\dot{\gamma}$ tends to $N\mathcal{H}_j \cap \mathcal{C}_\pm$ for exactly one value of $j$ (see Section \ref{sect:applications} for more on the connection with Theorem \ref{theo:bdynamics}).

 In preparation for the next result, recall the notation $\bar{H}^{s}(\mathcal{X})$ for the Sobolev space of $H^s$ distributions which are extendible across the boundary $\partial \mathcal{X}$, see \cite[Appendix B.2]{hormanderIII:1985}. If $\mathcal{M}$ is bounded by nondegenerate Killing horizons, let ${\varkappa}_\mathrm{sup}$ and ${\varkappa}_{\mathrm{inf}}$ denote the maximum and minimum of all $\varkappa_j$ over all $\mathcal{H}_j$.
 
 \begin{theo} \label{theo:vasy}
Suppose that $\mathcal{M}$ is bounded by nondegenerate Killing horizons and is nontrapping at zero energy. Let $\varkappa = {\varkappa}_{\mathrm{inf}}$ if $s \geq 1/2$, and $\varkappa = {\varkappa}_{\mathrm{sup}}$ if $s < 1/2$. Then,
\[
\widehat{\Box}(\sigma): \{u \in \bar{H}^{s}(\mathcal{X}): \widehat{\Box}(0)u \in \bar{H}^{s-1}(\mathcal X)\} \rightarrow \bar{H}^{s-1}(\mathcal{X})
\]
is Fredholm in the half-plane $\Im \sigma >(1/2-s)\cdot \varkappa$.
\end{theo}

\noindent Since $\mathcal{X}$ is spacelike, energy estimates for the wave equation imply that there exists $C_0$ such that $\widehat{\Box}(\sigma)$ is invertible for $\Im \sigma > C_0$ \cite{dyatlov:2011:cmp,hintz2013semilinear}. Analytic Fredholm theory shows that
\[
\widehat{\Box}(\sigma)^{-1} : C^\infty(\mathcal{X}) \rightarrow C^\infty(\mathcal{X})
\]
is a meromorphic family of operators for $\sigma \in \CC$. Quasinormal frequencies are then defined as poles $\widehat{\Box}(\sigma)^{-1}$, where the outgoing boundary conditions are recast as a smoothness condition up to $\partial \mathcal{X}$ (hence the use of extendible Sobolev spaces).

Theorem \ref{theo:vasy} was first established for Schwarzschild--de Sitter spacetimes by S\'a Baretto--Zworski  \cite{barreto1997distribution}, and for slowly rotating Kerr-de Sitter spacetimes by Dyatlov \cite{dyatlov:2011:cmp}. Vasy generalized these results, providing a robust proof as part of a more general microlocal framework \cite{vasy:2013}, and it is precisely this framework that we use to establish Theorem \ref{theo:vasy}.

A slightly more restrictive version of Theorem \ref{theo:vasy} was also established by Warnick \cite{warnick:2015:cmp} using physical space (rather than microlocal) methods. Roughly speaking, \cite{warnick:2015:cmp} applies when each Killing horizon is generated by $T$, and  $T$ is everywhere timelike in the interior $\mathcal{M}^\circ$ (in particular, the global nontrapping condition is vacuously satisfied).

Using Vasy's method, a version of Theorem \ref{theo:vasy} also holds for the Klein--Gordon equation on Kerr--AdS spacetimes \cite{gannot:2014:kerr} (Theorem \ref{theo:vasy} does not apply directly in that setting, since Kerr--AdS is not bounded by Killing horizons due to the presence of a conformally timelike boundary). Theorem \ref{theo:vasy} is also closely related to the meromorphic continuation of the resolvent on even asymptotically hyperbolic spaces --- see \cite[Chapter 5]{zworski:resonances}, \cite{zworski2015resonances}, and \cite[Section 4]{vasy:2013}.

\section{Preliminaries} \label{sect:prelims}

The purpose of this section is to recall some basic notions in Lorentzian b-geometry and fix the relevant notation. For a thorough treatment the reader is referred to \cite{melrose1993atiyah}.

\subsection{Manifolds with boundary}

Let $M$ be a smooth manifold with boundary $X$. Say that $x \in C^\infty(M)$ is a boundary defining function if
\[
x \geq 0, \quad X = \{ x = 0\}, \quad dx\neq 0 \text{ on $X$}.
\]
Local coordinates $(x,y^i)$ near $q\in X$ are called adapted coordinates provided $x$ is a local boundary defining function; Latin indices will always range $i=1,\ldots,n$. In particular, $(y^i)$ are local coordinates on $X$ when restricted to the boundary.

\subsubsection{The b-tangent bundle}

Those vector fields on $M$ which are tangent to $\partial M$ form a Lie subalgebra $\mathcal{V}_b(M)$ of $\mathcal{V}(M)$, the latter denoting the algebra of vector fields on $M$. In adapted coordinates $\mathcal{V}_b(M)$ is locally spanned by 
\begin{equation} \label{eq:btangentbasis}
x\partial_x\text{ and } \partial_{y^1},\ldots,\partial_{y^n}.
\end{equation}
Moreover, $\mathcal{V}_b(M)$ coincides with sections of a natural vector bundle $^b TM$, the b-tangent bundle. If $(a^0,a^i)$ are coefficients for a local basis \eqref{eq:btangentbasis}, then local coordinates on $^bTM$ are given by $(x,y^i,a^0,a^i)$.

The inclusion $\mathcal{V}_b(M) \hookrightarrow \mathcal{V}(M)$ induces a bundle map 
\begin{equation} \label{eq:i}
i : {^b T M} \rightarrow T M.
\end{equation}
Over points $q \in M^\circ$ this map is an isomorphism. On the other hand, if $q \in X$, then $i$ has a one-dimensional kernel denoted by ${^b N_q X}$. The subbundle $^b N X \subset {^bT_XM}$ is trivial: it is spanned by the nonvanishing section 
\[
x\partial_x \in \mathcal{V}_b(M).
\] 
This section is well defined independently of the choice of local boundary defining function $x$. The range of $i$ is $T_q X$, viewed as a subspace of $T_q M$. In summary, there is a short exact sequence
\begin{equation} \label{eq:btangentexact}
0 \rightarrow {^b NX} \rightarrow {^bT_X M} \rightarrow TX \rightarrow 0
\end{equation}
of bundles over $X$.

\subsubsection{The b-cotangent bundle} \label{subsubsect:bcotangent}

The dual bundle $^b T^*M = (^b TM)^*$ is called the b-conormal bundle; the pairing between $v \in {^bT_qM}$ and $\varpi \in {^bT^* _qM}$ will be written as $\left<v,\varpi\right>$. In adapted coordinates, sections of $^bT^*M$ are locally spanned by
\[
\frac{dx}{x} \text{ and } dy^1,\ldots,dy^n.
\]
If $(\sigma,\eta_i)$ are corresponding coefficients, then $(x,y^i,\sigma,\eta_i)$ are local coordinates on $^bT^*M$. Observe that $\sigma$ is actually a well defined function on $^b T^*_X M$, since $\sigma = \left<x\partial_x,\cdot\right>$ and $x\partial_x$ is defined over $X$ independently of the boundary defining function.

There is also a dual discussion with regards to the adjoint  $i^* : T^* M \rightarrow {^bT^* M}$ of \eqref{eq:i}. Again, over points $q \in M^\circ$ this map is an isomorphism. Over the boundary there is an injection
\begin{equation} \label{eq:adjointmap}
\pi: T^*X \rightarrow {^b T^*M},
\end{equation}
which is adjoint to the map $^b T M \rightarrow TX$ in \eqref{eq:btangentexact}.  Its image $^b T^*X$ consists of elements in $^b T^* M$ which annihilate $x\partial_x$.

At times it is convenient to fix a splitting of the exact sequence \eqref{eq:btangentexact}. This can be done smoothly by choosing a global boundary defining function $x$ and decomposing
\[
^b T_q M = {^bN_qX} \oplus \{ v \in {^bT_q M}: \left<v, dx/x\right> = 0 \}
\]
for $q\in X$. When no confusion can arise, denote the annihilator of $dx/x$ by $^b TX$, even though the latter subbundle depends on choices; then $i$ is an isomorphism from $^bTX$ onto $TX$.

If $Z \subset X$ is a submanifold and $q\in Z$, one may also consider the set of all vectors in $^b T_p M$ which are tangent to $Z$. The annihilator of this subspace is denoted by $^b N_q^* Z$. Observe that $^b N^*Z \subset {^bT^*_Z X}$ since elements of $^b N^*Z$ must annihilate $x\partial_x$. Viewed as a subbundle of $T^*_Z X$ via \eqref{eq:adjointmap}, the b-conormal bundle of $Z$ is canonically identified with the conormal bundle of $Z$ within $X$.

\subsection{Lorentzian b-metrics} \label{subsect:lorentzianbmetric}

A Lorentzian b-metric $g$ is a symmetric bilinear form on $^b TM$ of signature $(1,n)$. The inverse metric is denoted by $g^{-1}$. Write $G \in C^\infty({^bT^*M})$ for the associated quadratic form on $^b T^*M$: if $\varpi \in T^*_qM$, then
\[
G(\varpi) = g^{-1}(\varpi,\varpi).
\]
Given $v \in {^bT_qM}$, let $v^\flat \in {^bT^*_qM}$ denote the unique b-covector such that 
$g^{-1}(v^\flat,\cdot) = \left<v,\cdot\right>$. Similarly, if $\varpi \in {^b T_q^*M}$, then $\varpi^\sharp \in {^bT_qM}$ is the unique b-vector such that $g(\cdot,\varpi^\sharp) = \left<\cdot,\varpi\right>$ (so $df^\sharp$ is the b-gradient of $f$, but that terminology will not be used).

Consider the restriction of $G$ to the interior $^b T^* _{M^\circ}M$, which is canonically identified with $T^*M^\circ$. By restricting coordinates $(x,y^i,\sigma, \eta_i)$ for $^b T^*M$ to $T^* M^\circ$, the usual Hamilton vector field of $G$ is
\begin{equation} \label{eq:H_G}
H_G = (\partial_\sigma G)x\partial_x  - (x\partial_x G)\partial_\sigma + (\partial_{\eta_i}G)\partial_{y^i} - (\partial_{y^i} G)\partial_{\eta_i} 
\end{equation}
(for summation, the index $i$ is considered a subscript in $\partial_{y^i}$ and a superscript in $\partial_{\eta_i}$). This expression shows that $H_G$ has a unique extension to a vector field on $^bT^*M$ which is tangent to $^b T^*_XM$, as well as tangent to the level sets of $\sigma$ within $^bT^*_XM$.

The geodesic flow refers to the flow of $H_G$ on $^bT^*M\setminus 0$. Let
\[
\Sigma = \{ G = 0\} \setminus 0
\]
denote the union of light cones, which is invariant under $H_G$; the restriction of $H_G$ to $\Sigma$ generates the null-geodesic flow.

If $\widehat\rho\in C^\infty({^bT^*M}\setminus 0)$ is homogeneous of degree $-1$ in the fiber variables, then $\widehat\rho H_G$ descends to a well defined vector field on the quotient 
\[
^bS^*M = \left({^bT^*M} \setminus 0\right)/\RR_+,
\]
where $\RR_+$ acts by positive dilations in the fibers. Although the corresponding vector field on $^bS^*M$ depends on the choice of $\widehat{\rho}$, the integral curves of $\widehat{\rho} H_G$ on $^bS^*M$ are well defined up to reparametrization. Furthermore, integral curves of $H_G$ project down to reparametrizations of $\widehat{\rho}H_G$ integral curves; if $\widehat{\rho} > 0$ this preserves orientation as well. Since $\Sigma$ is conic, the null-geodesic flow is therefore equivalent to the $\widehat{\rho}H_G$ flow on $\widehat{\Sigma} = \Sigma/\RR_+ \subset {^bS^*M}$. Integral curves, on either $\Sigma$ or $\widehat{\Sigma}$, are called null-bicharacteristics.

\section{b-Horizons}

\subsection{Decomposing the metric} In this section $M$ will denote an arbitrary $(n+1)$ dimensional manifold with boundary $X$, equipped with a Lorentzian b-metric $g$. We first discuss the standard $1+n$ decomposition of the metric in the language of b-geometry. Assume there exists a boundary defining function $\tau$ for which $d\tau/\tau$ is timelike on $X$,
\[
G(d\tau/\tau) > 0.
\]
Define the lapse function $A = G(d\tau/\tau)^{-1/2}$ and set $N = -A(d\tau/\tau)^\sharp \in {^bTM}$ near $X$, so $g(N,N) = 1$. Causal vectors in $^bT M$ are oriented by declaring $N$ to be future directed. As in Section \ref{subsubsect:bcotangent}, decompose 
\[
^bT_XM = {^bNX} \oplus {^bTX},
\]
where $^bTX$ is the annihilator of $d\tau/\tau$. Equivalently $^bTX$ is the $g$-orthogonal complement of $\RR N$, so $g$ is negative definite on $^bTX$ since $N$ is timelike. Let 
\[
h = A^{2}\left(\dt\otimes \dt\right) - g.
\]
Here we are using $\otimes$ to denote a symmetric tensor product. Although this $(0,2)$-tensor is defined wherever $d\tau/\tau$ is timelike, only its properties over the boundary will be used. 
\begin{lem} \label{lem:hsignature} The tensor $h$ is positive definite on $^bTX$, and $\ker h = \RR N$.
\end{lem}
\begin{proof}
If $v \in {^bTX}$, then $h(v,w) = -g(v,w)$ for any $w \in {^bT_XM}$. In particular 
\[
h|_{^bTX} = -g|_{^bTX},
\] 
so $h$ is positive definite on $^bTX$. Additionally, the $h$-orthogonal complement of $^bTX$ equals its $g$-orthogonal complement, namely $\RR N$. Since $h(N,N) = 0$, the kernel of $h$ is $\RR N$.
\end{proof}
Define the shift vector $W = -(\tau\partial_\tau + AN) \in {^b T_X M}$. By construction $W$ is $g$-orthogonal to $N$, so $W \in {^bTX}$.
In adapted coordinates $(\tau,y^i)$,
\[
g|_X = A^2 \left(\dt\otimes \dt\right) - h_{ij}\left(dy^i - W^i\dt\right)\otimes\left(dy^j - W^j \dt\right).
\]
A similar construction holds for the inverse metric. Begin by defining the $(2,0)$-tensor
\[
k = N\otimes N - g^{-1}.
\]
In analogy with Lemma \ref{lem:hsignature}, consider the $g^{-1}$-orthogonal decomposition  
\[
^bT^*_XM = \RR d\tau/\tau\oplus (\RR d\tau/\tau)^\perp,
\] 
where $g^{-1}$ is negative definite on $(\RR d\tau/\tau)^\perp$. In general $(\RR d\tau/\tau)^\perp \neq  {^bT^*X}$, with equality only over those points where the shift vector vanishes.

\begin{lem}
	The tensor $k$ is positive definite on $^b T^*X$, and $\ker k = (\RR d\tau/\tau)^\perp$.
\end{lem}
\begin{proof}
	By definition $N$ annihilates $(\RR d\tau/\tau)^\perp$, so $k$ is positive definite on $(\RR d\tau/\tau)^\perp$. The $k$-orthogonal complement of $(\RR d\tau/\tau)^\perp$ is $\RR d\tau/\tau$, and the latter is the kernel of $k$. Although it is not necessarily true that $(\RR d\tau/\tau)^\perp = {^bT^*X}$, one always has that $d\tau/\tau \notin {^bT^*X}$, hence $k$ is positive definite on $^bT^*X$ as well.
\end{proof}

In adapted coordinates, the components $k^{ij}$ satisfy $k^{ij} = h^{ij}$. Here, $h^{ij}$ is the inverse of the matrix $h_{ij}$, which exists by Lemma \ref{lem:hsignature} since $\partial_{y^1},\ldots,\partial_{y^n}$ span $^bTX$. The corresponding expression for $g^{-1}|_X$ in these coordinates is 
\begin{equation} \label{eq:ginverse}
g^{-1}|_X = A^{-2}(\tau\partial_\tau + W^i \partial_i)(\tau\partial_\tau + W^i \partial_i) - h^{ij}\partial_i \partial_j.
\end{equation}
In coordinates $(\tau,y^i,\sigma,\eta_i)$,
\begin{equation} \label{eq:symbol}
G|_{{^bT^*_X M}} = A^{-2}\sigma^2  + 2A^{-2} W^i \eta_i \sigma + \left(A^{-2} W^i W^j - h^{ij}\right)\eta_i \eta_j.
\end{equation}
The derivatives of $G$ in the fiber variables are then
\begin{equation} \label{eq:HG}
\begin{cases}
\partial_{\eta_i} G = 2\left( (A^{-2} W^i W^j - h^{ij})\eta_j + A^{-2}W^i \sigma \right),\\
\partial_\sigma G = 2A^{-2}\left(\sigma + W^i \eta_i \right).
\end{cases}
\end{equation}
These are the coefficients of $\partial_{y^i}$ and $\tau\partial_{\tau}$, respectively, in the coordinate expression for $H_G$.

The first goal is to show that $dG \neq 0$ and $dG, d\tau$ are linearly independent on $\Sigma \cap {^bT^*_XM}$.  The former condition holds in a neighborhood of the boundary as well, so locally $\Sigma$ is an embedded conic submanifold of $^bT^*M\setminus 0$; no claim is made about the global behavior of $\Sigma$. The latter condition is that $\Sigma$ meets ${^bT^*_XM}$ transversally. Both statements follow from the timelike nature of $d\tau/\tau$, which implies that $G$ restricted to $^bT^*_XM$ is strictly hyperbolic in the direction of $d\tau/\tau$; explicitly,
\[
\partial_\sigma G \neq 0 \text{ on }\Sigma \cap {^bT^*_XM}.
\]

\begin{lem} \label{lem:dG}
$dG \neq 0$, and $dG, d\tau$ are linearly independent on $\Sigma \cap {^bT^*_XM}$.
\end{lem}
\begin{proof}
	If $\partial_\sigma G(\varpi) = 0$ for $\varpi \in {^bT^*_X M\setminus 0}$, then $g^{-1}(d\tau/\tau,\varpi) = 0$. Since $d\tau/\tau$ is timelike, this implies that $\varpi$ is spacelike, so $\varpi \notin \Sigma$. Therefore $\partial_\sigma G \neq 0$ and hence $dG \neq 0$ on $\Sigma \cap {^bT^*_XM}$. Clearly $dG$ and $d\tau$ are also therefore linearly independent on $\Sigma \cap {^b T^*_XM}$.
\end{proof}
As in Lemma \ref{lem:dG}, the set $\{g^{-1}(d\tau/\tau,\cdot) = 0 \}$ does not intersect $\Sigma$ near $^bT^*_X M$, since $d\tau/\tau$ is timelike there. At least near the boundary then, $\Sigma$ is the disjoint union $\Sigma = \Sigma_+ \cup \Sigma_-$, where
\[
\Sigma_\pm = \{ \pm g^{-1}(-d\tau/\tau,\cdot)  > 0 \}.
\]
Recall that timelike vectors are oriented with respect to $-d\tau/\tau$, so $\Sigma_\pm$ is the union of future/past-directed light cones. The sets $\Sigma_\pm$ are conic and invariant under $H_G$.

\subsection{b-Horizons and surface gravity} \label{subsect:horizon}

Here we discuss b-horizons in the sense of Definition \ref{defi:killing1}. In the stationary setting considered in Section \ref{subsect:stationary}, this definition corresponds to a Killing horizon generated by $T$. For more details, see Section \ref{subsect:relation}. Some examples are the event horizons of Schwarzschild and Schwarzschild--AdS spacetimes, as well as the event and cosmological horizons of the Schwarzschild--de Sitter spacetime. It excludes, however, their rotating (Kerr) generalizations, which are handled in Section \ref{subsect:generalhorizon}. This is also the setting which applies to the boundaries of even asymptotically hyperbolic metrics (after an appropriate change of smooth structure \cite[Section 4.9]{vasy:2013}).

Observe that the condition \eqref{eq:horizoncondition} is purely a statement over the boundary. In terms of the lapse function and shift vector, 
\[
\mu  = A^2 - h(W,W),
\]
and in adapted coordinates, \eqref{eq:surfacegravity} is equivalent to $\partial_{y^i} \mu = 2\varkappa W^j h_{ij}$ along $H$. If $\varkappa$ does not vanish, then $H\subset X$ is an embedded hypersurface and $\mu$ is a defining function.

Henceforth $H$ will denote an arbitrary nondegenerate b-horizon. Let $\omega \in {^bT^*X}$ denote the unique b-covector such that $\left<\cdot,\omega\right> = h(\cdot,W)$ on $^bTX$. Now 
\[
k(\omega,\omega) = h(W,W) = A^2- \mu,
\] 
so $\omega \neq 0$ near $H$. Since  $d\mu = 2\varkappa \omega$ on $H$, it follows that $\omega(q)$ at a point $q\in H$ spans $^bN^*_qH$. Let $\Pi$ be the $k$-orthogonal complement of the line bundle $\RR \omega$ within $^bT^*X$. Wherever $\omega \neq 0$, any $\varpi \in {^bT^*_X M}$ can be written uniquely as 
\begin{equation} \label{eq:varpi_decomp}
\varpi = \sigma(d\tau/\tau) + \xi \omega + \zeta,
\end{equation}
where $\xi \in \RR$ and $\zeta \in \Pi$ is the $k$-orthogonal projection of $\varpi - \sigma(d\tau/\tau)$ onto $\Pi$. In this parametrization,
\begin{equation} \label{eq:Ginomega}
G(\varpi) = A^{-2}\sigma^2 + 2A^{-2}  (A^2 -\mu) \xi \sigma - A^{-2}(A^2-\mu)\mu \xi^2 - k(\zeta,\zeta).
\end{equation}
The restriction of $g^{-1}$ to $\Pi$ is therefore negative definite, and the $g^{-1}|_{^bT^*X}$-orthogonal complement of $\Pi$ is $\RR\omega$. The signature of $g^{-1}|_{^bT^*X}$ changes as one passes from $\{ \mu  >  0 \}$ to $\{ \mu < 0\}$:
\begin{lem} \label{lem:grestrictedsignature}
The tensor $g^{-1}|_{^bT^*X}$ has the following properties.
\begin{enumerate} \itemsep6pt
	\item $g^{-1}|_{^bT^*X}$ is negative definite on $^bT^*X \cap \{ \mu  > 0\}$.
	\item  $g^{-1}|_{^bT^*X}$ is degenerate on $^bT^*_{H} X$ with kernel $ {^bN^*H}$.
	\item $g^{-1}|_{^bT^*X}$ is of Lorentzian signature on $^bT^*X \cap \{ \mu < 0 \}$.
\end{enumerate}
Therefore $\Sigma \cap {^b T^*X} \cap \{ \mu > 0 \} = \emptyset$, and furthermore $ \pm \xi < 0$ on $\Sigma_\pm \cap {^b T^*X}$.
\end{lem}
\begin{proof}
	Recall the $g^{-1}|_{^bT^*X}$-orthogonal decomposition $^bT^*X = \RR \omega \oplus \Pi$, 
	where $g^{-1}|_{^bT^*X}$ is negative definite on $\Pi$. The claims about the signature of $g^{-1}|_{^bT^*X}$ then follow from the expression
	\[
	G(\omega) = -A^{-2}(A^2-\mu)\mu,
	\]
	which changes sign upon crossing from $\{ \mu > 0\}$ to $\{ \mu < 0 \}$. This immediately shows that $\Sigma\cap {^bT^*X}$ does not enter the region $\{ \mu > 0\}$.  Additionally, 
	\[
	\partial_\sigma G|_{^bT^*X} = 2A^{-2}(A^2-\mu) \xi,
	\]
	so $\xi$ has the same sign as $\partial_\sigma G$ on $\Sigma \cap {^bT^*X}$.
\end{proof}

The same argument also determines the causal nature of $d\mu$ in a neighborhood of $H$. Of course $d\mu|_{H}$ is null by definition.

\begin{lem} \label{lem:dmucausal}
	There exists $\varepsilon > 0$ such that $d\mu$ is timelike on $\{ -\varepsilon < \mu < 0\}$ and spacelike on $\{ 0 < \mu < \varepsilon \}$.
\end{lem}
\begin{proof}
	Write $d\mu = \xi_0 \omega + \zeta_0$, where $\xi_0|_{H} = 2\varkappa$ and $\zeta_0|_{H} = 0$. Since the surface gravity is nonvanishing and $k(\zeta_0,\zeta_0)$ vanishes quadratically along $H$,
	\begin{equation} \label{eq:G(dmu)}
	G(d\mu) = -A^{-2}(A^2-\mu)\mu\xi_0^2 - k(\zeta_0,\zeta_0)
	\end{equation}
	has the opposite sign as $\mu$ in a neighborhood of $H$. By compactness of $H$, this neighborhood can be chosen of the form $\{ | \mu | < \varepsilon\}$ for some $\varepsilon > 0$. \end{proof}

Next consider the location of $\Sigma \cap {^bT^*_X M}$ for $\sigma \neq 0$. The main observation here is that only one component can enter $\{\mu \geq 0\}$.

\begin{lem} \label{lem:onecomponent}
If $\pm \sigma >0$, then $\Sigma_\pm \cap \left(^bT^*X + \sigma \dt\right)\cap \{\mu \geq 0\} = \emptyset$.
\end{lem}
\begin{proof}
Write $\varpi \in \Sigma$ as  $\varpi = \sigma(d\tau/\tau) + \eta$ with $\eta \in {^bT^*X}$. Then
\[
\sigma \partial_\sigma G(\varpi) = A^{-2}\sigma^2 - G(\eta).
\]
If $\mu \geq 0$, then $G(\eta) \leq 0$ according to Lemma \ref{lem:grestrictedsignature}. Thus $\sigma \neq 0$ and $\mu \geq 0$ implies $\partial_\sigma G$ has the same sign as $\sigma$ on $\Sigma$.
\end{proof}

\subsection{Hamilton vector field} \label{subsect:flow} We now begin our study of the null-geodesic flow.
\begin{lem} \label{lem:HGmu}
	The restriction of $H_G \mu$  to $ \Sigma \cap {^bT^*_H M}$ vanishes precisely at ${^bN^*H}$. Furthermore, 
	\[
	\pm (\sgn \varkappa)(H_G \mu)(\varpi) < 0
	\]
	for $\varpi \in \left(\Sigma_\pm \cap {^bT^*_H M} \right) \setminus {^bN^*H}$.
	 
\end{lem}
\begin{proof}
	Suppose that $\varpi \in \Sigma \cap {^bT^*_H M}$ satisfies $(H_G \mu) (\varpi)= 0$. Since $H_G \mu = 2g^{-1}(d\mu,\cdot)$, it follows that $\varpi$ and $d\mu$ are orthogonal null b-covectors, hence are collinear. Since $d\mu$ is proportional to $\omega$ on $H$, this implies that $\varpi \in {^b N^*H}$. Additionally, $\omega \in \Sigma_-$ over $H$ according to Lemma \ref{lem:grestrictedsignature}, so $(\sgn \varkappa) d\mu \in \Sigma_-$ over $H$. If $\varpi \in \left(\Sigma_\pm \cap {^bT^*_H M} \right) \setminus {^bN^*H}$, then $\pm  \varpi$ and $(\sgn \varkappa) d\mu$ are not orthogonal and lie in opposite light cones.
\end{proof}

\noindent As a consequence of Lemma \ref{lem:HGmu}, the set ${^bN^*H}$ is invariant under $H_G$. Indeed, 
\[
^bN^*H = \{\tau = 0,\, \sigma =0,\, G = 0, \, \mu = 0\},
\]
and $H_G$ annihilates these functions on $^bN^*H$. The two halves $\mathcal{R}_\pm = \Sigma_\pm \cap {^bN^*H}$ are also invariant under $H_G$. Furthermore, for each $C > 0$ there exists $\varepsilon > 0$ such that 
\[
\pm (\sgn \varkappa)(H_G \mu) < 0
\]
on a conic set ${^bT^*_X M} \cap \Sigma_\pm  \cap \{ |\sigma| > C^{-1}|\partial_\sigma G| \}\cap  \{ |\mu| < \varepsilon \}$. This can be improved in the region $\{ \mu < 0\}$ where $d\mu$ is timelike.

\begin{lem} \label{lem:HGmunegative}
If $\varpi \in {^bT^*_X M} \cap \Sigma_\pm \cap \{\mu < 0\}$ then $\pm (\sgn \varkappa) H_G \mu < 0$. 
\end{lem}
\begin{proof}
	Since $d\mu$ is timelike in the region of interest, $H_G\mu$ cannot vanish. But $(\sgn \varkappa)d\mu \in \Sigma_-$ over $H$ as noted in Lemma \ref{lem:HGmu}, so the sign condition holds on $\{ \mu < 0\}$, at least for $|\mu|$ sufficiently small.
\end{proof}

\noindent Thus $\pm(\sgn \varkappa) \mu$ is always decreasing strictly along null-bicharacteristics emanating from ${^bT^*_X M} \cap \Sigma_\pm \cap \{\mu < 0\}$.

\begin{lem} \label{lem:HGradial} If $\eta \in {^bN^*H}$, then $H_G(\eta) = 2\varkappa \xi \left(\eta_i \partial_{\eta_i}\right)$, where $\eta = \eta_i dy^i  = \xi \omega$.
\end{lem}
\begin{proof}
	From \eqref{eq:Ginomega}, it is easy to see that $\left(\partial_{\eta_i} G\right)(\eta)= 0$. On the other hand,
	\[
	\left(\partial_{y^i} G\right)(\eta) = -\xi^2 \partial_{y^i} \mu = -2\varkappa \xi^2 \omega_i,
	\]
 where $\omega_i = \left<\partial_{y^i}, \omega \right>$. Since $\eta_i = \xi \omega_i$, the proof follows.
\end{proof}

\noindent Therefore $H_G$ is radial on $^bN^*H\setminus 0$ in the sense that is a multiple of the vector field $\eta_i \partial_{\eta_i}$ generating dilations in the fibers of $^bT^*M$. Define 
\[\widehat{\rho} = |2g^{-1}(-d\tau/\tau,\cdot)|^{-1},
\] 
which is a degree $-1$ homogeneous function, well defined near $\mathcal{R}_\pm$. Equivalently $\widehat{\rho} = |\tau^{-1}H_G\tau|^{-1}$. Thus $\widehat{\rho}H_G$ descends to a vector field on $\widehat{\Sigma}_\pm$. Let 
\[
L_\pm = \mathcal{R}_\pm / \RR_+ \subset \widehat{\Sigma}_\pm.
\]
That $H_G$ is radial at $^bN^*H\setminus 0$ corresponds to $\widehat{\rho} H_G$ vanishing at $L_\pm$. 
Passing to the quotient disregards whether $H_G$ was originally pointing towards or away from the zero section; this is encoded in $H_G \widehat{\rho}$, which can be viewed as a function on $^bS^*M$.
\begin{lem} \label{lem:H_Grho}
The function $\widehat{\rho}$ satisfies $H_G \widehat{\rho}\,|_{L_\pm} = \pm\varkappa$.
\end{lem}
\begin{proof}
	Within the boundary, $\widehat{\rho} = |\partial_\sigma G|^{-1}$. In that case,
	\[
	H_G \widehat{\rho} = -\frac{H_G(\partial_\sigma G)}{(\partial_\sigma G)^2}\cdot \sgn(\partial_\sigma G).
	\]
	Now $\left(\partial_\sigma G\right)(\eta) = 2\xi$, and since this is a linear function of $\eta$,
	\[
	(H_G (\partial_\sigma G))(\eta)  = 4\varkappa \xi^2
	\]
	at $\eta \in {^bN^*H}$ by Lemma \ref{lem:HGradial}, so $H_G \widehat{\rho}  \,|_{L_\pm} =\pm\varkappa$ as claimed.
\end{proof}	

In light of Lemma \ref{lem:H_Grho}, define $\beta = |H_G \widehat{\rho}\,|$, viewed as a function on $^bS^*_X M$ defined near $L_\pm$. Thus $H_G\widehat{\rho} = \pm(\sgn \varkappa)\beta$ in a sufficiently small neighborhood of $L_\pm$, and if $\varkappa$ is constant (hence can be viewed as a function on $^bS^*_X M$ rather than just over $H$), then $\beta = |\varkappa|$ near $L_\pm$ modulo functions vanishing along $L_\pm$.

\subsection{The null-geodesic flow near $L_\pm$} \label{subsect:nearLpm}
Since $\widehat{\rho}H_G$ vanishes at $L_\pm$, it is natural to consider stability properties of the null-geodesic flow near these sets. It is shown that $L_\pm$ is either a source or a sink for the null-geodesic flow (in fact for the entire geodesic flow) within $^bS^*_X M$, depending on the sign of $\varkappa$ and the subscript $\pm$. 

{\bf Thus, unless stated otherwise, for the remainder of this section $\widehat{\rho}H_G$ will be considered by restriction as a vector field on $^bS^*_X M$, rather than $^bS^*M$}.

To begin, recall the decomposition \eqref{eq:varpi_decomp} and define $K: {^b T^*_X} M \rightarrow \RR$ by $K(\varpi) = k(\zeta,\zeta)$, as in \eqref{eq:Ginomega}. Set
\[
\widehat{\sigma} = \widehat{\rho}\sigma, \quad \widehat{K} = \widehat{\rho}^{\,2} K;
\]
these descend to functions near $L_\pm \subset {^b}S^*_XM$. Since $\zeta = 0$ on $^b N^*H$, it follows that $L_\pm$ is defined by 
\[
L_\pm= \{\tau =0, \, \mu =0,\, \widehat \sigma = 0,\, \widehat{K} = 0\}.
\]
As general notation, if $U \subset {^b}T^*_X M$, then $\mathcal{I}(U)$ will denote the ideal of smooth functions on $^bT^*_X M$ which vanish along $U$; the same notation will be used if $U \subset {^b}S^*_X M$. In particular, $\mathcal{I}(L_\pm)$ is the ideal of smooth functions on $^bS^*_XM$ which vanish along $L_\pm$.

\begin{lem} \label{lem:H_Gsigma}
	The function $\widehat{\sigma}$ satisfies $\widehat{\rho} H_G \widehat{\sigma} = \pm (\sgn \varkappa) \beta \widehat{\sigma} $ near $L_\pm$.
\end{lem}
\begin{proof}
	Since $H_G \sigma = 0$ identically when restricted to $^bT^*_X M$, it follows that $\widehat{\rho} H_G \widehat{\sigma} = \widehat{\sigma} H_G \widehat{\rho}$ on $^bT^*_X M$. By definition, this equals $\pm(\sgn\varkappa)\beta \widehat{\sigma}$ near $L_\pm$ . 
\end{proof}
For the next result, recall the decomposition $d\mu = \xi_0 \omega + \zeta_0$ as in Lemma \ref{lem:dmucausal}, where $\xi_0|_{H} = 2\varkappa$ and $\zeta_0|_H = 0$.
\begin{lem} \label{lem:H_Gmuplusigma}
	The function $\mu \pm 4\hspace{.05em}\widehat{\sigma}$ satisfies 
	\[
	\widehat{\rho}H_G (\mu \pm 4\hspace{.05em}\widehat{\sigma}) = \pm 2(\sgn \varkappa)\beta (\mu \pm 4\hspace{.05em}\widehat{\sigma}) + \mathcal{I}(L_\pm)^{2}
	\] 
	near $L_\pm$.
\end{lem}
\begin{proof}
	One way to see this is to write
	\begin{equation} \label{eq:H_Gmualternate}
	H_G\mu =2A^{-2}(A^2-\mu)\xi_0\sigma - 2A^{-2}(A^2-\mu)\mu \xi \xi_0 - 2k(\zeta,\zeta_0).
	\end{equation}
	 Since $\xi_0|_{H} = 2\varkappa$ along $H$, it follows that $\xi_0 = 2(\sgn \varkappa)\beta + \mathcal{I}(L_\pm)$ near $L_\pm$ when lifted to $^bS^*_XM$. But $-2\widehat{\rho}\xi = \pm1$ at $L_\pm$, so after multiplication by $\widehat{\rho}$ this yields 
	\[
	\widehat{\rho}H_G \mu = 2(\sgn \varkappa)\beta(\pm \mu + 2\widehat{\sigma}) + \mathcal{I}(L_\pm)^2
	\] 
	near $L_\pm$. It remains to combine this with Lemma \ref{lem:H_Gsigma}.
\end{proof}

Finally, consider the action of $\widehat{\rho}H_G$ on $\widehat{K}$. By construction, the vanishing of $\zeta$ defines the subbundle $\mathrm{span}(d\tau/\tau,\omega)\subset {^b}T^*_X M$. Let 
\begin{equation} \label{eq:Omega}
\Omega = \left(\mathrm{span}(d\tau/\tau,\omega)\setminus0\right)/\RR_+ \subset {^b}S^*_XM.
\end{equation}
Then $\widehat{K}$ has nondegenerate quadratic vanishing along $\Omega$. It is now convenient to work locally, where one may choose a $k$-orthonormal basis of sections $(\omega,E^A)$ for $^bT^*X$. Throughout, capitalized indices will always range over $A = 1,\ldots,n-1$. If 
\[
\varpi = \sigma(d\tau/\tau)+ \xi \omega + \zeta_A E^A,
\] 
then $K(\varpi) = \delta^{AB}\zeta_A\zeta_B$. In canonical coordinates, $\zeta_A = a^{i}_A\eta_i$ where $a^{i}_A = k(E^A,dy^i)$. Therefore $a^{i}_A \partial_{y^i}$ is tangent to $H$, since $E^A$ is $k$-orthogonal to $d\mu$ along $H$. \begin{rem}
	Even though $\zeta_A$ is a priori only a (locally defined) function on $^bT^*_XM$, its Hamilton vector field can still be defined by choosing an arbitrary extension to $^bT^*M$; by tangency, its restriction to $^bT^*_XM$ is independent of the extension.
\end{rem}
 In particular, $H_{\zeta_A}$ as a vector field on $^bT^*_XM$ is just
\begin{equation} \label{eq:HzetaA}
H_{\zeta_A} = a^i_A \partial_{y^i} - \partial_{y^i}(a_A^j)\eta_j \partial_{\eta_i},
\end{equation}
which is tangent to $^bT^*_H X$.

\begin{lem} \label{lem:H_Ggamma}
The function $\widehat{K}$ satisfies 
\[
\widehat{\rho}H_G \widehat{K} = \pm 2(\sgn \varkappa)\beta \widehat{K}+ \mathcal{I}(S^*_H X)\cdot \mathcal{I}(\Omega)
\]
near $L_\pm$.
\end{lem}
\begin{proof}
	Since $\widehat{\rho}H_G \widehat{K}  = \pm 2(\sgn \varkappa)\beta \widehat{K} + \widehat{\rho}^{\,3} H_G K$, it suffices to examine $H_G K$. Write $G = G_0 - K$, where 
	\[
	G_0 = A^{-2}\sigma^2 + 2A^{-2}(A^2-\mu)\xi\sigma - A^{-2}(A^2 - \mu)\mu \xi^2.
	\]
	Since $^bT^*_H X =\{\tau = 0, \mu = 0, \sigma = 0\}$, it follows that $G_0 \in \mathcal{I}(^bT^*_H X)$. Working locally in the notation of the preceding paragraph,
	\[
	H_G K = H_{G_0} K = -2\delta^{AB} \zeta_A H_{\zeta_B} G_0.
	\]
	The tangency of $H_{\zeta_B}$ to $^bT^*_H X$ implies that $\widehat{\rho}^{\,2} H_{\zeta_B} G_0 \in \mathcal{I}(^bS^*_H X)$, whereas $\widehat\zeta_A \in \mathcal{I}(\Omega)$. 
	\end{proof}

Let $\rho_\pm = \digamma^{-1}\widehat K + (\mu \pm 4\hspace{.05em}\widehat{\sigma})^2 + \widehat{\sigma}^2$, viewed as a function on $^bS^*_X M$. Here $\digamma>0$ is a yet undetermined constant which will be taken large. The vanishing of $\rho_\pm$ near $L_\pm$ defines $L_\pm$, and moreover $\rho_\pm$ attains a non-degenerate minimum along $L_\pm$. Indeed, the quadratic form $(\mu\pm 4\widehat{\sigma})^2+ \widehat{\sigma}^2$ is positive definite in $(\mu,\widehat{\sigma})$, and $\digamma^{-1}\widehat{K}$ is positive definite in $\widehat{\zeta}$.

\begin{proof}[Proof of Theorem \ref{theo:bdynamics}] Combining Lemmas \ref{lem:H_Gsigma}, \ref{lem:H_Gmuplusigma}, \ref{lem:H_Ggamma} above shows that
	\[
	\pm (\sgn\varkappa)\widehat{\rho}H_G\rho_\pm \geq 2\beta \rho_\pm + \digamma^{-1} \mathcal{I}(S^*_HX)\cdot \mathcal{I}(\Omega) + \mathcal{I}(L_\pm)^3.
	\]
	A function in $\mathcal{I}(S^*_HX)\cdot \mathcal{I}(\Omega)$ is a finite sum of the form $\sum_k f_k \cdot g_k$,
	where $f_k \in  \mathcal{I}(S^*_HX)$ and $g_k \in \mathcal{I}(\Omega)$. By Cauchy--Schwarz and compactness, for any $\gamma > 0$ there is $C_\gamma > 0$ such that
	\[
	\digamma^{-1}\sum_k |f_k g_k| \leq \digamma^{-1}\left(\gamma \widehat{K} + C_\gamma(\mu^2 + \widehat{\sigma}^2)\right).
	\]
	If $\gamma>0$ is first chosen sufficiently small and $\digamma >0$ is subsequently taken to be sufficiently large, then this error can be absorbed by $2\beta \rho_\pm$, at the expense of replacing the coefficient with $2(1-\delta)\beta$ for any $\delta > 0$. Similarly, the error term in $\mathcal{I}(L_\pm)^3$ can be absorbed by the main term by restricting to a sufficiently small (possibly $\delta$-dependent) neighborhood of $L_\pm$. Thus for any $\delta > 0$,
\begin{equation} \label{eq:quantitativesourcesink}
\pm(\sgn \varkappa) \widehat{\rho}H_G \rho_\pm \geq 2(1-\delta)\beta \rho_\pm
\end{equation}
near $L_\pm$. This shows that $L_+$ is a source and $L_-$ is a sink for the $\widehat{\rho}H_G$ flow within $^bS^*_X M$ if $\varkappa > 0$; when $\varkappa < 0$ the source/sink behavior is reversed. If $\varkappa$ is constant, then $\beta$ in \eqref{eq:quantitativesourcesink} can be replaced with $|\varkappa|$.
\end{proof}
As remarked in the paragraph following the statement of Theorem \ref{theo:bdynamics}, the global behavior of the $\widehat{\rho}H_G$ flow (not necessarily restricted to $^bS^*_X M$) is more complicated, and depends on the sign of $\varkappa$: at $L_\pm$,
\[
 \left(\tau^{-1}\widehat{\rho}H_G\tau\right)|_{L_\pm} = \left(\widehat{\rho}\,\partial_\sigma G\right)|_{L_\pm} = \mp 1,
\]
so infinitesimally, there is a stable direction transverse to $^bS^*_X M$ at $L_+$, and an unstable direction at $L_-$.

\subsection{More general b-horizons} \label{subsect:generalhorizon}

Definition \ref{defi:killing1} of a b-horizon corresponds in the stationary setting to a Killing horizon generated by $T$. But already in the stationary case it is important to consider more general Killing vector fields $K$, as these arise naturally for rotating spacetimes. We can treat this case as well, but only under serious additional symmetry hypotheses (see the last item in Definition \ref{def:bhorizonrotating}).

 In general, $\tau\partial_{\tau}$ is replaced with a section $K$ of the affine space $\tau \partial_{\tau} + {^b}TX$ (keeping in mind that ${^b}TX$ is defined relative to a choice of boundary defining function). Given $K$, define $\mu = \mu_K = g(K,K)$. If $K = \tau\partial_{\tau} + V$ for some section $V$ of ${^b}TX$, then
\[
\mu = A^2 - h(W-V,W-V),
\]
where $W$ is the usual shift vector. When appropriate, $V$ will be considered as a vector field on $X$.
\begin{defi} \label{def:bhorizonrotating}
A subset $H \subset X$ is called a {\bf b-horizon generated by $K = \tau\partial_{\tau} + V$} if $H$ is a compact connected component of $\{\mu=0\}$ satisfying the following conditions.
\begin{enumerate} \itemsep6pt

	\item There is a function $\varkappa : H \rightarrow \RR$, called the {\bf surface gravity}, such that
	\begin{equation} \label{eq:Vhorizon}
	d\mu^\sharp = 2\varkappa K \text{ on } H.
	\end{equation}

	\item $V\mu$ vanishes quadratically along $H$,
	
	\item $V\varkappa = 0$ along $H$,
	
	\item There is a Riemannian metric $\underline{\ell}$ on $H$ such that $V|_H$ is Killing with respect to $\underline\ell$.

\end{enumerate}
	If $\varkappa$ is nowhere vanishing, then the b-horizon generated by $K$ is said to be {\bf nondegenerate}.
\end{defi}

Assume that $H$ is a nondegenerate b-horizon generated by $K$. The spatial components of \eqref{eq:Vhorizon} in adapted coordinates read
\[
\partial_{y^i}\mu = 2\varkappa (W^j - V^j)h_{ij}.
\]
Since $W-V \neq 0$ near $H$ and the surface gravity never vanishes, $H$ is embedded and $W-V$ is a nonzero normal to $H$. 

\begin{rem} \begin{inparaenum}[(1)] \item 
If $V$ is identically zero, then when expressed covariantly, the $d\tau/\tau$ component of \eqref{eq:Vhorizon} is trivial. More generally, the $d\tau/\tau$ component of $K^\flat$ satisfies
\[
\left<\tau\partial_\tau,K^\flat\right> = A^2 - h(W,W-V) 
\] 
on $H$, so if this vanishes, then $h(V,W-V) = 0$. Thus \eqref{eq:Vhorizon} implies that $V$ is automatically tangent to $H$, but we require the stronger condition that $V\mu$ actually vanishes quadratically along $H$. 

\item The existence of $\underline{\ell}$ is a nontrivial restriction. Since $H$ is compact and connected, the existence of such a metric is equivalent to the statement that $V|_{H}$ is an $\RR$-linear combination of mutually commuting vector fields on $H$ with periodic flows \cite{lynge1973sufficient}.  \end{inparaenum}
\end{rem}

Let $\omega \in {^bT^*X}$ be dual to $W - V \in {^b}TX$ with respect to $h$. As before, any $\varpi \in {^bT^*_X M}$ can be written uniquely as $\varpi = \sigma (d\tau/\tau) + \xi \omega + \zeta$, where $\zeta \in \Pi$ is $k$-orthogonal to $\omega$. Then,
\begin{equation} \label{eq:GinomegaV}
G(\varpi) = A^{-2}(\sigma+\eta_V)^2 + 2A^{-2}  (A^2 -\mu) \xi (\sigma+\eta_V) - A^{-2}(A^2-\mu)\mu \xi^2 - k(\zeta,\zeta),
\end{equation}
where $\eta_V : {^b}T^*_X M \rightarrow \RR$ is given by $\eta_V = \left< V, \cdot\right>$.
\begin{rem}
In this section it will sometimes be convenient to choose adapted coordinates $(\tau, y^i)$ near $H$ such that $y^n = \mu$ and $k(dy^A,\omega) = 0$, recalling $A = 1,\ldots,n-1$. In other words, $dy^A \in \Pi$; this is possible since the vector field $W-V$ dual to $\omega$ is transverse to $H$.
\end{rem}

  From \eqref{eq:GinomegaV} it is clear that $\eta_V$ plays an important role.
 Note that $\eta_V$ vanishes at $^bN^*H$, since $d\mu$ spans $^bN^*H$ over $H$ and $\eta_V(d\mu) = V\mu$, with the latter vanishing on $H$.

\begin{lem} \label{lem:etaVproperties} 
	The function $\eta_V$ has the following properties.
	\begin{enumerate} \itemsep6pt
		\item $H_{\eta_V}$ is tangent to $^bN^*H$ and $^bT^*_H X$.
		\item $H_{\eta_V}\xi$ vanishes along $^bN^*H$.
		\item The restriction of $H_{G}\eta_V$ to $^bS^*_XM$ vanishes quadratically along $^bN^*H$.
	\end{enumerate}
\end{lem}
\begin{proof} \begin{inparaenum}[(1)] \item To see that $H_{\eta_V}$ is tangent to $^bT^*_H X$, note that $H_{\eta_V} \sigma$ vanishes on $^bT^*_XM$, and $H_{\eta_V}\mu = \eta_V(d\mu) = V\mu$, which vanishes at $^bT^*_H X$. Now write
		\[
		^bN^*H = \{\tau = 0,\, \sigma =0,\,\mu=0,\,\zeta =0\}.
		\] 
		Since $H_{\eta_V}$ is tangent to $^bT^*_H X$, it follows from \eqref{eq:GinomegaV} that $H_G \eta_V$ vanishes along $^bN^*H$, noting that $k(\zeta,\zeta)$ vanishes quadratically there. Now observe that
		\[
		^bN^*H = \{\tau = 0,\, \sigma =0,\, \mu = 0, \, G = 0,  \, \eta_V = 0\},
		\]
		hence $H_{\eta_V}$ is tangent to $^bN^*H$.
		
		\item  Because $H_{\eta_V}$ is tangent to $^bT^*_HX$, in order to compute $H_{\eta_V}\xi$ it suffices to first restrict $\xi$ to $^bT^*_HX$. But in canonical coordinates $(\tau,y^i,\sigma,\eta_i)$ with $y^n = d\mu$ and $dy^A \in \Pi$,
		\[
		\xi|_{^bT^*_HX} = 2\varkappa \eta_n.
		\]
		Under the assumption that $V\mu$ vanishes quadratically along $H$ and $V\varkappa=0$, this function is annihilated by $H_{\eta_V}$ at $^bN^*H$. To see this, observe that $\eta_A = 0$ at $^bN^*H$ for $A=1,\ldots,n-1$, so
		\[
		H_{\eta_V} = V^i \partial_{y^i} -\partial_{y^i}(V^j)\eta_j \partial_{\eta_i} = V^i \partial_{y^i} - \partial_{y^i}(V^n)\eta_n \partial_{\eta_i}
		\]
		at $^bN^*H$. But $V^n = V\mu$ vanishes quadratically along $H$, so the second term is zero. On the other hand, the first term annihilates $2\varkappa \eta_n$ since $V\varkappa = 0$.
		
		\item From what has been shown and \eqref{eq:GinomegaV}, \eqref{eq:HzetaA}, it follows that $H_{G}\eta_V$ vanishes quadratically along $^bN^*H$ within $^bS^*_XM$. \end{inparaenum}\end{proof}

Lemmas \ref{lem:grestrictedsignature}, \ref{lem:onecomponent} do not hold as stated, although versions are true if the hyperplane $^b T^*X \subset {^bT^*_X M}$ is replaced by $\{ \sigma = -\eta_V\}$. In particular, $\Sigma_\pm \cap {^bT^*_X  M}$ can intersect $\{\mu > 0\}$. In addition, $H_G$ is no longer radial at $^bN^*H\setminus 0$:

\begin{lem} \label{lem:HGnotradial} Assume that $V\mu$ vanishes quadratically along $H$. If $\eta \in {^bN^*H}$, then $H_G(\eta) = 2 \xi \left(\varkappa \,\eta_i \partial_{\eta_i} + V^i \partial_{y^i}\right)$, where $\eta = \eta_i dy^i  = \xi \omega$.
\end{lem}
\begin{proof}
	Recall that $\eta_V$ vanishes at $^bN^*H$, so $(\sigma+\eta_V)^2$ vanishes quadratically there. Thus relative to Lemma \ref{lem:HGradial}, there is an additional term
	\[
	2\xi\left(\partial_{\eta_i}(\eta_V)\partial_{y^i} -  \partial_{y^i}(\eta_V)\partial_{\eta_i} \right)
	\]
	at $^bN^*H$. As in Lemma \ref{lem:etaVproperties}, the second term vanishes at $^bN^*H$.
\end{proof}

Even though $H_G$ is not radial there, $^bN^*H$ is still invariant under the $H_G$ flow. To see this, recall from Lemma \ref{lem:etaVproperties} that
\begin{equation} \label{eq:conormaldefinedby}
^bN^*H = \{\tau = 0,\, \sigma =0,\, \mu = 0, \, G = 0,  \, \eta_V = 0\}.
\end{equation}
But $d\mu|_H$ is null, so $H_G\mu = 0$ at $^bN^*H$. It then suffices to note that $H_G {\eta_V}=0$ at $^bN^*H$ by Lemma \ref{lem:etaVproperties}.

\begin{rem}
 Lemma \ref{lem:dmucausal} holds true in this more general setting. Indeed $G(d\mu)$ has the same sign as $-\mu$ since $V\mu$ vanishes quadratically at $H$. Also observe that $\partial_{\sigma} G = 2\xi$ at $^bN^*H$, so $\pm \xi < 0$ on $\mathcal{R}_\pm$, and $(\sgn\varkappa)d\mu \in \Sigma_-$ over $H$. Therefore Lemmas \ref{lem:HGmu}, \ref{lem:HGmunegative} are still true.
\end{rem}

The sets $L_\pm$ are defined as before, and it is subsequently shown that $L_\pm$ is still either a source or sink for the $\widehat{\rho}H_G$ flow within $^bS^*_X M$, completing the proof of Theorem \ref{theo:bdynamics} in this more general setting. 

{\bf  As in Section \ref{subsect:nearLpm}, for the remainder of this section $\widehat{\rho}H_G$ will be considered by restriction as a vector field on $^bS^*_X M$, rather than $^bS^*M$}.

\begin{lem} \label{lem:H_GrhoV} If $V \varkappa =0$, then $\widehat{\rho}$ satisfies $H_G \widehat{\rho}\,|_{L_\pm} = \pm \varkappa$.
\end{lem}
\begin{proof}
	Again, work in canonical coordinates $(\tau,y^i,\sigma,\eta_i)$, where $y^n = \mu$. By tangency to $^bN^*H$ and Lemma \ref{lem:HGnotradial}, relative to Lemma \ref{lem:H_Grho} there is an additional contribution to $H_G (\partial_\sigma G)$ of the form
	\[
	4\xi V^i \partial_{y^i} (\varkappa)\eta_n
	\]
at $^bN^*H$, cf. Lemma \ref{lem:etaVproperties}. This vanishes since $V\varkappa = 0$.
\end{proof}

As before, let $\beta = |H_G \widehat{\rho}|$. The calculation of $\widehat{\rho}H_G \widehat{\sigma}$ in Lemma \ref{lem:H_Gsigma} goes through unchanged, namely 
\[
\widehat{\rho}H_G \widehat{\sigma} = \pm(\sgn\varkappa)\beta\widehat{\sigma}
\]
near $L_\pm$. However, Lemmas \ref{lem:H_Gmuplusigma}, \ref{lem:H_Ggamma} both need to be modified. As a preliminary, define $\widehat{\eta}_V = \widehat{\rho}\eta_V$. Then, 
\[
H_G\eta_V=  \pm (\sgn \varkappa) \beta\widehat{\eta}_V +\mathcal{I}(L_\pm)^{2}
\] 
since $H_G \eta_V$ vanishes quadratically (within $^bS^*_X M$) along $^b N^*H$ by Lemma \ref{lem:etaVproperties}.

\begin{lem}
The function $\mu \pm 4\widehat{\sigma} \pm 4\widehat{\eta}_V$ satisfies 
\[
\widehat{\rho}H_G (\mu \pm 4\widehat{\sigma} \pm 4\widehat{\eta}_V) =  \pm2(\sgn\varkappa)\beta (\mu \pm 4\widehat{\sigma} \pm 4\widehat{\eta}_V) + \mathcal{I}(L_\pm)^2
\]
near $L_\pm$.
\end{lem}

\begin{proof}
	Compared to \eqref{eq:H_Gmualternate}, there is an additional term in $H_G\mu$ of the form
	\[
	2A^{-2}(\sigma+\eta_V + (A^2-\mu)\xi)(V\mu) + 2A^{-2}(A^2-\mu) \eta_V \xi_0.
	\]
	After multiplication by $\widehat{\rho}$, the first term is in $\mathcal{I}(L_\pm)^{2}$, while the second term equals $4(\sgn\varkappa)\beta \widehat{\eta}_V$ modulo $\mathcal{I}(L_\pm)^{2}$. Thus $\widehat{\rho}H_G \mu = 2(\sgn\varkappa)\beta (\pm \mu + 2\widehat{\sigma} + 2\widehat{\eta}_V) + \mathcal{I}(L_\pm)^2$,
	so 
	\[
	\widehat{\rho}H_G (\mu \pm 4\widehat{\sigma} \pm 4\widehat{\eta}_V) =  \pm2(\sgn\varkappa)\beta (\mu \pm 4\widehat{\sigma} \pm 4\widehat{\eta}_V) + \mathcal{I}(L_\pm)^{2}
	\] 
	near $L_\pm$.
\end{proof}

It remains to localize in the remaining fiber directions; for this, it is no longer appropriate to use the function $K$ as in Section \ref{subsect:nearLpm}. Instead, fix any Riemannian metric $\ell$ on $X$ such that 
\[
\ell|_H = d\mu^2 + \underline{\ell},
\]
where $\underline\ell$ is provided by Definition \ref{def:bhorizonrotating}; identify $\ell$ with a fiber metric on $^bTX$. If $(y^i)$ are local coordinates near $H$ with $y^n = \mu$, then
\[
(\ell|_H)_{AB} = \underline{\ell}_{AB}, \quad (\ell|_H)^{AB} = \underline{\ell}^{AB},
\]
where $\underline{\ell}_{AB}$ and $\underline{\ell}^{AB}$ are the components of $\underline{\ell}$ and its inverse with respect to the coordinates $(y^A|_{H})$ on $H$.

Now $\Pi$ is replaced with $\underline{\Pi}$, the $\ell^{-1}$-orthogonal complement of $d\mu$ within $^bT^*X$.  Let $\underline{\zeta}$ be the $\ell^{-1}$-orthogonal projection of $\eta \in {^b}T^*X$ onto $\underline{\Pi}$. The replacement for $K$ is the function $L:{^b}T^*_X M \rightarrow \RR$ defined by
\[
L(\varpi) = \ell^{-1}(\underline{\zeta},\underline{\zeta}).
\]
Similar to the remark preceding Lemma \ref{lem:etaVproperties}, it is possible to choose adapted coordinates $(\tau,y^i)$ where $y^n = \mu$ and $dy^A \in \underline{\Pi}$. The orthogonal projection onto $\underline{\Pi}$ is given in these coordinates by $\underline\zeta = \eta_A dy^A$. In particular,
\begin{equation} \label{eq:Lincoords}
L = \ell^{AB}\eta_A\eta_B.
\end{equation}
The following  lemma exploits that $V|_H$ is Killing with respect to $\underline{\ell}$. Define $\underline{\Omega}$ as in \eqref{eq:Omega}, only replacing $\omega$ with $d\mu$.

\begin{lem} \label{lem:H_Ggammarotating}
	The function $\widehat{L}$ satisfies 
	\[
	\widehat{\rho}H_G \widehat{L} = \pm 2(\sgn \varkappa)\beta \widehat{K}+ \mathcal{I}(S^*_H X)\cdot \mathcal{I}(\underline{\Omega}) + \widehat{\eta}_V \cdot \mathcal{I}(\underline{\Omega}) + \mathcal{I}(L_\pm)^3
	\]
	near $L_\pm$.
\end{lem}
\begin{proof}
	Begin by writing $G = G_0 + G_1 - K$, where now 
	\[
	\begin{cases}
	G_0 = A^{-2}\sigma^2 + 2A^{-2}(A^2-\mu)\xi\sigma - A^{-2}(A^2 - \mu)\mu \xi^2,\\
	G_1 = 2A^{-2}\sigma \eta_V + A^{-2}\eta_V^2 + 2A^{-2}(A^2-\mu)\xi\eta_V.
	\end{cases}
	\]
By choosing local orthonormal frames for $\Pi$ and $\underline{\Pi}$, it may be assumed that locally
\[
K = \delta^{AB}\zeta_A \zeta_B  \quad L = \delta^{AB} \underline{\zeta}_A \underline{\zeta}_B.
\]
Exactly as in Lemma \ref{lem:H_Ggamma}, one has $\widehat{\rho}^{\,3}H _{G_0} L \in \mathcal{I}(S^*_H X)\cdot \mathcal{I}(\underline{\Omega})$, since $H_{\underline{\zeta}_A}$ is tangent to $^bT^*_HX$ (see \eqref{eq:HzetaA}). Similarly, $\widehat{\rho}^{\,3}H_{-K}{L} \in \mathcal{I}(L_\pm)^3$. To see this, it suffices to show that $H_{\zeta_A}\underline{\zeta}_B \in \mathcal{I}(N^*H)$. Write
\[
\zeta_A = a^i_A \eta_i, \quad \underline{\zeta}_B = \underline{a}^i_A \eta_i.
\]
Since these functions are linear in the fibers, $H_{\zeta_A}\underline{\zeta}_B$ is the commutator of $a^i_A\partial_{y^i}$ with $\underline{a}^i_B\partial_{y^i}$ paired with $\eta_idy^i$. But both vector fields are tangent to $H$, so the commutator vanishes when paired with $d\mu \in {^b}N^*H$.

 Finally, consider $\widehat{\rho}^{\,3}H_{G_1}L$. Modulo errors as in the statement of the lemma, this comes down to evaluating $\widehat{\rho}^{\,2}H_{\eta_V} L$, corresponding to the third summand in $G_1$. Using local coordinates as in the preceding paragraph and $\eqref{eq:Lincoords}$,
\[
\widehat{\rho}^{\,2}H_{\eta_V} L = V^i \partial_{y^i} (\ell^{AB}) \widehat\eta_A \widehat\eta_B - 2\partial_{y^A}(V^i) \ell^{AB}\widehat\eta_B\widehat\eta_i.
\]
Here, $\widehat{\eta}_A = \widehat{\rho}\eta_A$. When $i=n$ the expression above vanishes cubically along $L_\pm$. On the other hand, the sum over the remaining indices $i=1,\ldots,n-1$ is exactly the Lie derivative of $\underline{\ell}$ with respect to $V|_H$ evaluated at $\widehat{\eta}_A dy^A$, hence vanishes cubically along $L_\pm$. The proof is now complete, noting that $\widehat{\rho}H_G \widehat{L} = \pm 2 (\sgn \varkappa)\beta \widehat{L} + \widehat{\rho}^{\,3}H_G L$ near $L_\pm$.
\end{proof}

Let $\widehat{\rho}_\pm = \digamma^{-1}\widehat{L}+(\mu\pm4\widehat{\sigma}\pm4\widehat{\eta}_V)^2 + \widehat{\sigma}^2 + \widehat{\eta}_V^{\,2}$ near $L_\pm$, where $\digamma> 0$ will once again be chosen large. Since the quadratic form $(\mu\pm4\widehat{\sigma}\pm 4\widehat{\eta}_V)^2 + \widehat{\sigma}^2 + \widehat{\eta}_V^{\,2}$ is positive definite in $(\mu,\widehat{\sigma},\widehat{\eta}_V)$, it follows that $\widehat{\rho}_\pm$ attains its minimum along $L_\pm$. Now consider the proof of Theorem \ref{theo:bdynamics} in this setting.

\begin{proof} [Proof of Theorem \ref{theo:bdynamics}]
The proof is essentially the same as in Section \ref{subsect:nearLpm}; the only difference is an additional error term of the form $\digamma^{-1} \widehat{\eta}_V \cdot f$, where $f \in \mathcal{I}(\underline{\Omega})$ is independent of $\digamma>0$. This can then be bounded by 
\[
\digamma^{-1}|\widehat{\eta}_V \cdot f| \leq \digamma^{-1}(\gamma \widehat{L} + C_\gamma \widehat{\eta}_V^{\,2}).
\] 
Since $\widehat{\rho}_\pm$ is actually positive definite in $(\mu,\widehat{\sigma},\widehat{\eta}_V)$, the last term can be absorbed by choosing $\digamma>0$ sufficiently large.
\end{proof}

\section{Applications to stationary spacetimes and quasinormal modes} \label{sect:applications}

In this section we outline how Vasy's method \cite{vasy:2013} is used to establish Theorem \ref{theo:vasy} for stationary spacetimes bounded by Killing horizons. As remarked in the introduction, one of the original applications of \cite{vasy:2013} was to the proof of Theorem \ref{theo:vasy} for the particular case of Kerr-de Sitter spacetimes.

\subsection{Relationship with stationary spacetimes} \label{subsect:relation} Let $\mathcal{M}$ be a stationary spacetime, identified globally with $\RR_t \times \mathcal{X}$. Define 
\[
M = [0,\infty)_\tau\times \mathcal{X},
\]
and embed $\mathcal{M}\subset M$ via the map $(t,q) \mapsto (e^{-t},q)$. Then the boundary $X$ of $M$ is $\{\tau=0\}\times \mathcal{X}$, where $\tau = e^{-t}$ is a boundary defining function (note that we are being somewhat imprecise with the terminology when $\mathcal{M}$ itself has a boundary --- in that case $M$ is a manifold with corners). We think of $M$ as a compactification of $\mathcal{M}$ obtained by gluing in the component $X$ located at future infinity.

Since $dt = -d\tau/\tau$ and the original metric is stationary, $g$ admits a unique extension to a Lorentzian b-metric on $M$. This extension is invariant under positive dilations in $\tau$, as seen from the equality $T = -\tau\partial_{\tau}$. Because $d\tau/\tau$ is everywhere timelike, the sets $\Sigma_\pm$ are well defined throughout $^bT^*M$, and the function $\widehat{\rho}$ is nonvanishing on $\Sigma_\pm$.

Now consider the $\widehat{\rho}H_G$ flow on $^bT^*M \cap \Sigma_\pm$, recalling that $\widehat{\rho} = |\tau^{-1}H_G\tau|^{-1}$. In canonical coordinates $(\tau,y^i,\sigma,\eta_i)$,
\[
\widehat{\rho}H_G = \mp \tau\partial_{\tau} + (\widehat{\rho}\,\partial_{\eta_i} G)\partial_{y^i} - (\widehat{\rho}\,\partial_{y^i} G)\partial_{\eta_i}.
\]
On $^bT^*_{M^\circ}M \cap \Sigma_\pm = T^*\mathcal{M}  \cap \Sigma_\pm$, this is just a rescaling of the usual null-geodesic flow on $T^*\mathcal{M}$. Furthermore
\[
t(\exp(s\widehat{\rho}H_G)\varpi) = t_0\pm s,
\]
where $\varpi \in T^*\mathcal{M}  \cap \Sigma_\pm$ satisfies $t(\varpi) = t_0$. The flow of $\widehat{\rho}H_G$ on $^bT^*_XM \cap \Sigma_\pm$ is therefore equivalent to the quotient of the (rescaled) null-geodesic flow on $T^*\mathcal{M} \cap \Sigma_\pm$ by time translations.

Now suppose that $\mathcal{M}$ is a spacetime bounded by nondegenerate Killing horizons which is nontrapping at zero energy. Let $M$ be its compactification as above, and identify $\partial X$ with $\partial \mathcal{M} \cap \mathcal{X}$. Zero-energy null-geodesics correspond to the $\widehat{\rho}H_G$ flow on $T^*\mathcal{M} \cap \Sigma_\pm \cap \{\sigma =0\}$. In fact, based on the remarks in the previous paragraph, that $\mathcal{M}$ is nontrapping at zero energy is equivalent to the statement that for each \[
\varpi \in \big( {^b}S^*X \cap \widehat{\Sigma}_\pm \big) \setminus {^b}SN^*\partial X
\] 
the following two conditions are satisfied:
\begin{enumerate} \itemsep6pt 
	\item  $\exp(s\widehat{\rho}H_G)\varpi\rightarrow {^b}SN^*\partial X$ as $s\rightarrow \mp \infty$,
	\item there exists $s_0\geq 0$ such that $\exp(\pm s_0 \widehat{\rho}H_G) \varpi \in {^b S^*_{\partial X}} X$.
\end{enumerate}


\subsection{The extended spacetime} Suppose that $\mathcal{M}$ is bounded by Killing horizons and is nontrapping at zero energy. In order to exploit the full dynamical structure, it is important to embed $\mathcal{M}$ in a larger stationary spacetime $\mathcal{N}$ obtained by extension across the horizons. For ease of notation, it will be assumed in this section that $\partial \mathcal{M}$ consists of a single connected component; the general case is handled analogously.

By stationarity, it suffices to enlarge $\mathcal{X}$. To begin, embed $\mathcal{X}$ in a boundaryless compact manifold $\widetilde{\mathcal{X}}$, and extend $g$ arbitrarily to a stationary metric on $\RR \times \widetilde{\mathcal{X}}$. Also recall that 
\[
K = T + V,
\] 
where $V$ is tangent to $\mathcal{X}$ and satisfies $[T,V]=0$. In view of this last condition, $V$ is tangent to each time slice. Extend $V$  to a $T$-invariant vector field on $\RR \times \widetilde{\mathcal{X}}$, which therefore defines an extension of $K$ (which is of course not Killing in general). If $\mu = g(K,K)$, then $T\mu =0$,
so $\mu$ can be identified with a function on $\widetilde{\mathcal{X}}$. Since the surface gravity is assumed to be nonvanishing, $\partial  \mathcal{X} \subset \widetilde{\mathcal{X}}$ is an embedded hypersurface, and if $\varepsilon > 0$ is sufficiently small, then $\{ |\mu| \leq \varepsilon\}$ defines a neighborhood of $\partial \mathcal{X}$ in $\widetilde{\mathcal{X}}$ diffeomorphic to $[-\varepsilon,\varepsilon] \times \partial\mathcal{X}$. With this identification one has $[0,\varepsilon]\times \partial \mathcal{X} \subseteq \mathcal{X}$, since by assumption $K$ is timelike in a punctured neighborhood of $\mathcal{H}$ within $\mathcal{M}$.

Choose $\varepsilon > 0$ sufficiently small so that $d\mu$ is timelike in $\{-\varepsilon\leq \mu < 0\}$. With this choice, let $\mathcal{Y}\subset \widetilde{\mathcal{X}}$ be the union
\[ 
\mathcal{Y} = \mathcal{X} \cup \{|\mu| \leq \varepsilon\},
\]
and finally set $\mathcal{N} = \RR \times \mathcal{Y}$. Observe that $\mathcal{Y}$ has a boundary component corresponding to $\{\mu = -\varepsilon\} \times \partial \mathcal{X}$. Next, set 
\[
N = [0,\infty) \times \mathcal{Y}, \quad Y= \{\tau =0\} \times \mathcal{Y}.
\]  
Equip $N$ with its induced Lorentzian b-metric as in Section \ref{subsect:relation}. Then, $\partial \mathcal{X}$ can be identified with a hypersurface $H \subset Y$, and the claim is that $H$ is a b-horizon in the sense of Definition \ref{def:bhorizonrotating}. 

To begin, note that the  $T$-invariant vector field $K = T + V$ induces a section of $\tau\partial_{\tau} + {^b}TY$. Indeed, $T = -\tau\partial_{\tau}$, and the tangency condition $Vt= -\left<V,d\tau/\tau\right> = 0$ extends by continuity from $\mathcal{N}$ to show that $V$ defines a section of $^bTY$.
\begin{lem}
	$H \subset Y$ is a nondegenerate b-horizon generated by $-K = \tau\partial_{\tau} - V$.
\end{lem}
\begin{proof}
\begin{inparaenum}[(1)] The four hypotheses in Definition \ref{def:bhorizonrotating} are verified. \item Recall that $\grad_g \mu = -2\varkappa K$ on $\mathcal{H}$, where the gradient is with respect to the original metric on $\mathcal{M}$. Since $T$ is Killing, a brief calculation shows that $T\varkappa = 0$.
By continuity from $\mathcal{N}$, it follows that $d\mu^\sharp = -2\varkappa K$ on $H$. 
	
	\item Since  $V$ is Killing on $\mathcal{M}$, it is easy to see that $V \mu =0$ on the interior $\mathcal{M}^\circ$. Although the extension of $V$ to $\mathcal{N}$ is not assumed to be Killing, it nevertheless follows by continuity from $\mathcal{M}^\circ$ that $V \mu$ vanishes to infinite order along $\mathcal{H}$, and hence along $H$ by $T$-invariance.
	
	\item As in the first item of the proof, a simple calculation using that $V$ is Killing on $\mathcal{M}$ shows that $V\varkappa=0$.

\item Let $\underline{\ell}$ be the positive definite metric on $\partial \mathcal{X}$ induced by $g$. Now $V$ is tangent to $\partial \mathcal{X}$, and by definition $V$ generates isometries with respect to $\underline{\ell}$. It then suffices to consider $\underline{\ell}$ as a metric on $H$ via the latter's identification with $\partial \mathcal{X}$.
\end{inparaenum}	
\end{proof}

Now consider the trapping properties of null-geodesics in the extended region; recall that 
\[
L_{\pm} = {^bSN^*}H \cap \widehat \Sigma_\pm.
\]
Since $\varkappa > 0$ in this setting (see Definition \ref{defi:bounded}), it follows that $L_{\pm}$ is a source/sink within $^bS^*_{Y}N$.

\begin{lem} \label{lem:nontrappingextended}
	Let $\mathcal{M}$ be nontrapping at zero energy. If $	\varpi \in \big( {^b}S^*Y \cap \widehat{\Sigma}_\pm \big)\setminus L_\pm$, then the following conditions hold:
	\begin{enumerate} \itemsep6pt 
		\item $\exp(s\widehat{\rho}H_G)\varpi \rightarrow L_{\pm}$ as $s\rightarrow \mp \infty$,
		\item there exists $s_0 \geq 0$ such that $\mu(\exp(\pm s_0\widehat{\rho}H_G)\varpi) = -\varepsilon$.
	\end{enumerate} 
\end{lem}
\begin{proof}
	Let $\varpi$ be as above.\begin{inparaenum}[(1)]\item If $\varpi \in {^b}S^*_X Y$, then it follows from Definition \ref{defi:nontrapping} that $\exp(s \widehat \rho H_G)\varpi \rightarrow L_{\pm}$ as $s\rightarrow \mp\infty$.  So now assume that $\varpi \in {^b}S^*_{Y\setminus X}Y$, in the region obtained by extension across $H$. Since $d\mu$ is timelike in the extended region by construction, $\pm \mu(\exp(s\widehat\rho H_G)\varpi)$ increases as $s$ decreases by Lemma \ref{lem:HGmunegative}. If $\mu(\exp(s_0\widehat{\rho}H_G)\varpi) = 0$ for some finite value $\pm s_0 < 0$, then Definition \ref{defi:nontrapping} applies. Otherwise, $\mu(\exp(s \widehat{\rho}H_G)\varpi) \rightarrow \mu_0 \leq 0$ as $s\rightarrow \mp\infty$. By the mean value theorem, there is a sequence $s_n \rightarrow \mp \infty$ along which
	\begin{equation} \label{eq:muzero}
	\mu(\exp(s_n \widehat{\rho}H_G)\varpi) \rightarrow \mu_0, \quad (H_G\mu)(\exp(s_n \widehat{\rho}H_G)\varpi) \rightarrow 0.
\end{equation}
It follows by compactness that $\exp(s_n \widehat \rho H_G)\varpi \rightarrow \varpi^\star$ for some $\varpi^\star \in \widehat{\Sigma}_\pm$ (recall that $\widehat{\Sigma}$ is closed), extracting a subsequence if necessary. In particular,
\[
\mu(\varpi^\star) = \mu_0, \quad (H_G \mu)(\varpi^\star) = 0,
\]
which by Lemma \ref{lem:HGmunegative} implies that $\mu_0 = 0$. On the other hand, if $\mu_0 = 0$, then it follows from Lemma \ref{lem:HGmu} that $\varpi^\star \in L_\pm$. But $L_{\pm}$ is a source/sink, which implies that $\exp(s \widehat{\rho} H_G)\varpi \rightarrow L_{\pm}$ as $s\rightarrow \mp\infty$, not just along the sequence $s_n$.

%
%
	
	\item By Definition \ref{defi:nontrapping} there is no loss in first assuming that $\mu (\varpi) \leq 0$. In fact, since $\varpi \notin L_\pm$ it may be assumed that $\mu(\varpi) < 0$ by Lemma \ref{lem:HGmu}. Once this is known, it follows that $\mu(\exp(\pm s_0\widehat{\rho}H_G)\varpi) = -\varepsilon$ for some $s_0 \geq 0$ by Lemma \ref{lem:HGmunegative} and the same argument as in the first part of the proof.
\end{inparaenum}
\end{proof}

Lemma \ref{lem:nontrappingextended} shows that nontrapping properties are unaffected by the extension procedure, precisely due to the switch in causality of $d\mu$.

\subsection{Outline of the proof of Theorem \ref{theo:vasy}} The extended stationary wave operator on $\mathcal{Y}$ will also be denoted by $\widehat{\Box}(\sigma)$. Under the assumption that $\mathcal{M}$ is nontrapping at zero energy, all the hypotheses of \cite[Sections 2.2, 2.6]{vasy:2013} are satisfied. In the notation there, we take 
\[
\tilde{\rho} = \widehat{\rho}, \quad \rho_0 = \rho_\mp, \quad \beta_0 = \varkappa, \quad \beta_1 = \beta, \quad \beta_{\mathrm{inf}} = (\varkappa_{\mathrm{sup}})^{-1}, \quad \beta_{\mathrm{sup}} = (\varkappa_{\mathrm{inf}})^{-1},
\] 
suppressing the subscript $j$ in $\rho_0, \, \beta_0, \, \beta_1$ which indexes the different boundary components. However, as compared to \cite[Section 2]{vasy:2013}, we do not use a complex absorbing operator in the extended region $\mathcal{Y} \setminus \mathcal{X}$. Instead, we use the fact that $\widehat{\Box}(\sigma)$ is strictly hyperbolic with respect to the appropriate $d\mu_j$ in $\mathcal{Y} \setminus \mathcal{X}$, see \cite[Remark 2.6]{vasy:2013} and Lemma \ref{lem:dmucausal}. For a detailed discussion of this alternative see \cite{hintz2013semilinear,gannot:2014:kerr,zworski2015resonances}, as well as the textbook treatment in \cite[Chapter 5]{zworski:resonances}. Vasy's method then shows that
\[
\widehat{\Box}(\sigma): \{u' \in \bar{H}^{s}\left(\mathcal{Y}\right): \widehat{\Box}(0)u' \in \bar{H}^{s-1}\left(\mathcal Y\right)\} \rightarrow \bar{H}^{s-1}\left(\mathcal{Y}\right)
\]
is Fredholm of index zero in the half-plane $\Im \sigma >(1/2-s)\cdot \varkappa$. 

The difference between this result and Theorem \ref{theo:vasy} is that $\mathcal{X}$ has been replaced by the larger space $\mathcal{Y}$ (note that the extension is essentially arbitrary). To relate this back to the original problem on $\mathcal{X}$, first choose a continuous linear extension $e_s: \bar{H}^{s-1}(\mathcal{X}) \rightarrow \bar{H}^{s-1}(\mathcal{Y})$ and let $R \subset \bar{H}^{s-1}(\mathcal{Y})$ denote the range of $\widehat{\Box}(\sigma)$. Since $R$ has finite codimension in $\bar{H}^{s-1}(\mathcal{Y})$, the image of $e_s$ intersected with $R$ is of finite codimension in the image of $e_s$. But $e_s$ is injective, so there is a finite codimension subspace $S\subset  \bar{H}^s(\mathcal{X})$ such that $e_s(S) \subset R$. If $f \in S$, then
\[
\widehat{\Box}(\sigma)u' = e_s f
\]
has a solution $u' \in \bar{H}^{s}\left(\mathcal{Y}\right)$. Restricting $u'$ gives a solution to $\widehat{\Box}(\sigma)u=f$ on $\mathcal{X}$. Therefore 
\begin{equation} \label{eq:unextendedoperator}
\widehat{\Box}(\sigma): \{u \in \bar{H}^{s}\left(\mathcal{X}\right): \widehat{\Box}(0)u \in \bar{H}^{s-1}\left(\mathcal X\right)\} \rightarrow \bar{H}^{s-1}\left(\mathcal{X}\right)
\end{equation}
has closed, finite codimensional range for $\Im \sigma >(1/2-s)\cdot \varkappa$. Now the same energy estimates that imply $\widehat{\Box}(\sigma)$ is invertible on $\mathcal{Y}$ for large $\Im \sigma>0$ imply invertibility on $\mathcal{X}$ as well. Since the index of semi-Fredholm operators is locally constant, the operator \eqref{eq:unextendedoperator} is Fredholm of index zero.

\section*{Acknowledgments}
This paper is partially based upon work supported by the National Science Foundation under Grant No. 1502632. The author would like to thank Peter Hintz for his interest in the problem and for several valuable comments and suggestions.

	\bibliographystyle{plain}
	
	\bibliography{biblio}

\end{document}